\pdfoutput=1
\RequirePackage{ifpdf}
\ifpdf 
\documentclass[pdftex]{sigma}
\else
\documentclass{sigma}
\fi

\numberwithin{equation}{section}

\newtheorem{Theorem}{Theorem}[section]
\newtheorem*{Theorem*}{Theorem}
\newtheorem{Corollary}[Theorem]{Corollary}
\newtheorem{Lemma}[Theorem]{Lemma}
\newtheorem{Proposition}[Theorem]{Proposition}
 { \theoremstyle{definition}

\newtheorem{Example}[Theorem]{Example}
\newtheorem{Remark}[Theorem]{Remark} }

\usepackage{mathrsfs}
\usepackage{tikz}
\usepackage[all,cmtip]{xy}
\usepackage{bm}

\newcommand{\pd}{\partial}

\newcommand{\bC}{{\mathbb C}}

\newcommand{\bP}{{\mathbb P}}

\newcommand{\bZ}{{\mathbb Z}}

\newcommand{\cB}{{\mathcal B}}

\newcommand{\cF}{{\mathcal F}}

\newcommand{\half}{\frac{1}{2}}

\newcommand{\tC}{{\widetilde C}}

\newcommand{\tG}{{\widetilde G}}

\newcommand{\hf}{{\hat f}}

\DeclareMathOperator{\Aut}{Aut}
\DeclareMathOperator{\Id}{id}

\newcommand{\be}{\begin{equation}}
\newcommand{\ee}{\end{equation}}
\newcommand{\bea}{\begin{eqnarray}}
\newcommand{\ben}{\begin{eqnarray*}}
\newcommand{\een}{\end{eqnarray*}}
\newcommand{\eea}{\end{eqnarray}}

\begin{document}
\allowdisplaybreaks

\newcommand{\arXivNumber}{2210.08712}

\renewcommand{\PaperNumber}{085}

\FirstPageHeading

\ShortArticleName{Diagonal Tau-Functions of 2D Toda Lattice Hierarchy}

\ArticleName{Diagonal Tau-Functions of 2D Toda Lattice Hierarchy,\\ Connected $\boldsymbol{(n,m)}$-Point Functions, and Double\\ Hurwitz Numbers}

\Author{Zhiyuan WANG~$^{\rm a}$ and Chenglang YANG~$^{\rm b}$}

\AuthorNameForHeading{Z.~Wang and C.~Yang}

\Address{$^{\rm a)}$~School of Mathematics and Statistics, Huazhong University of Science and Technology,\\
\hphantom{$^{\rm a)}$}~Wuhan, P.R.~China}
\EmailD{\href{mailto:wangzy23@hust.edu.cn}{wangzy23@hust.edu.cn}}

\Address{$^{\rm b)}$~Hua Loo-Keng Center for Mathematical Sciences, Academy of Mathematics\\
\hphantom{$^{\rm b)}$}~and Systems Science, Chinese Academy of Sciences, Beijing, P.R.~China}
\EmailD{\href{mailto:yangcl@amss.ac.cn}{yangcl@amss.ac.cn}}

\ArticleDates{Received December 18, 2022, in final form October 21, 2023; Published online November 04, 2023}

\Abstract{We derive an explicit formula for the connected $(n,m)$-point functions associated to an arbitrary diagonal tau-function $\tau_f(\bm t^+,\bm t^-)$ of the 2d Toda lattice hierarchy using fermionic computations and the boson-fermion correspondence. Then for fixed $\bm t^-$, we compute the KP-affine coordinates of $\tau_f(\bm t^+,\bm t^-)$.
As applications, we present a unified approach to compute various types of connected double Hurwitz numbers, including the ordinary double Hurwitz numbers, the double Hurwitz numbers with completed $r$-cycles, and the mixed double Hurwitz numbers. We also apply this method to the computation of the stationary Gromov--Witten invariants of $\bP^1$ relative to two points.}

\Keywords{2d Toda lattice hierarchy; connected $(n,m)$-point functions; boson-fermion cor\-respondence; double Hurwitz numbers}

\Classification{37K10; 14N10; 14N35}

\section{Introduction}

\subsection{Double Hurwitz numbers}

Hurwitz numbers~\cite{hur} count the numbers of branched covers between Riemann surfaces
with specified ramification types.
They relate the geometry of Riemann surfaces to many other mathematical theories
such as the representation theory,
integrable hierarchies,
and combinatorics.
In particular,
Hurwitz numbers play an important role in the intersection theory on moduli spaces of curves
and Gromov--Witten theory,
see, e.g.,~\cite{bems, bm, dij, elsv1, elsv2, gv, op2, op, pa}.

The (possibly disconnected) ordinary Hurwitz numbers count all branched covers
between not necessarily connected Riemann surfaces,
and can be calculated using representation theory and the Burnside formula,
see, e.g.,~\cite{dij}.
In mathematical physics,
sometimes it is more natural to consider connected Hurwitz numbers.
For example,
the famous ELSV formula~\cite{elsv1, elsv2}
relates the connected single Hurwitz numbers to some Hodge integrals over the moduli spaces of stable curves
\cite{dm, kn},
and implies the polynomiality of such Hurwitz numbers.

Our main objects of interest in this work are several types of connected double Hurwitz numbers,
and the relative stationary Gromov--Witten invariant of $\mathbb{P}^1$.
The simplest example of double Hurwitz numbers is the
ordinary double Hurwitz number labeled by two partitions $\mu^\pm$,
which counts the branched covers which has ramification type $\mu^+$, $\mu^-$ over two given points
and simple ramifications over other points.
In~\cite{Ok1},
Okounkov showed that the generating series of all (possibly disconnected)
double Hurwitz numbers is a tau-function of the 2d Toda lattice hierarchy,
and found a fermionic representation of this tau-function:
\begin{equation*}
	\tau^{(2)} ( \bm t^+ , \bm t^-; \beta) =
	\big\langle 0 \big| \Gamma_+(\bm t^+) {\rm e}^{\beta K^{(2)}} \Gamma_- (\bm t^-) \big|0\big\rangle,
\end{equation*}
where $K^{(2)}$ is the cut-and-join operator.
In~\cite{gjv}, Goulden--Jackson--Vakil proved the piecewise polynomiality of the connected double Hurwitz numbers
using a purely combinatorial method.
Roughly speaking,
the whole affine space \big(with coordinates the parts of $\mu^\pm$\big)
is separated into some chambers by some walls,
and the piecewise polynomiality means that these numbers are polynomials in the parts of $\mu^\pm$
inside each chamber.
In~\cite{jo},
Johnson derived a formula for connected double Hurwitz numbers in each chamber,
and proved the strong piecewise polynomiality for ordinary double Hurwitz numbers using this formula.

There are also double Hurwitz numbers of other types in literatures.
For example,
the double Hurwitz numbers with completed $r$-cycles~\cite{op, ssz},
monotone and mixed double Hurwitz numbers~\cite{ggn}.
In~\cite{ssz},
Shadrin, Spitz, and Zvonkine derived a formula for double Hurwitz numbers with completed $r$-cycles
(see~\cite[equation~(17)]{ssz})
and proved the strong piecewise polynomiality using a method similar to the method in~\cite{jo}.
The monotone Hurwitz numbers were introduced by Goulden, Guay-Paquet, and Novak in~\cite{ggn2}
to formulate the Harish--Chandra--Itzykson--Zuber matrix model,
which are also attractive in many mathematical researches.
And a more complicated kind of Hurwitz numbers called the mixed double Hurwitz numbers were introduced in~\cite{ggn}
in the study of the combinatorial aspects of Cayley graphs of the symmetric groups.
The generating series of the mixed double Hurwitz numbers is also a tau-function of the 2-Toda hierarchy~\cite{ggn}.

\subsection{Motivation}

In this work,
we derive an explicit formula for connected $(n,m)$-point functions of
a diagonal tau-function~\cite{ca} of the 2d Toda lattice hierarchy,
and apply it to compute the connected double Hurwitz numbers
(both in chambers and on walls)
and the stationary GW invariants of $\mathbb{P}^1$ relative to two points.

This paper is part of a series of works~\cite{jwy, wy, wy1, wyz}
in which the fermionic approach to integrable hierarchies
are applied to solve problems in Gromov--Witten type theories.
These works are inspired by Zhou~\cite{zhou1}.
In that work, Zhou derived a formula for the connected bosonic $n$-point functions
of a tau-function of the KP hierarchy in terms of the KP-affine coordinates on the Sato Grassmannian.
See~\cite{jm, djm, sa, sw} for the basics of the boson-fermion correspondence and Sato's theory
of integrable hierarchies,
and see~\cite{by, hb, zhou3} for an introduction of the KP-affine coordinates
and the application to the Witten--Kontsevich tau-function~\cite{kon, wit}.
Inspired by Zhou's work on KP hierarchy,
we have derived formulas to compute the connected bosonic~$n$-~or $(n,m)$-point functions
for other integrable hierarchies
including the BKP hierarchy~\cite{djkm, jm} and diagonal tau-functions of $2$-BKP hierarchy,
see~\cite{wy} and~\cite{wy1} respectively.
Moreover,
in~\cite{jwy} the authors have developed a strategy to find the quantum spectral curve of type $B$
in the sense of Gukov--Su{\l}kowski~\cite{gs} using BKP-affine coordinates,
and computed the quantum spectral curve for spin Hurwitz numbers~\cite{eop, gkl}.
In~\cite{wyz}, the same method have been applied to find the quantum spectral curve of type $B$
for the generalized Br\'ezin--Gross--Witten models.
Now the present paper is devoted to the computation of the free energy of a diagonal tau-functions of $2$-Toda hierarchy.

In the case of KP (resp.~BKP) hierarchy,
the information of a tau-function $\tau$ is encoded entirely in its KP- (resp.~BKP-) affine coordinates,
and finding affine coordinates is equivalent to expressing the tau-function
as a Bogoliubov transform of the fermionic vacuum using only fermionic creators,
see~\cite{wy, zhou1}.
However,
in the case of $2$-Toda~\cite{ut} or $N$-component KP hierarchy~\cite{kv} (for a general $N\geq 2$),
such Bogoliubov transforms are not unique and yet we do not know a canonical way to
specify a set of coordinates.
Nevertheless,
in the case of diagonal tau-functions of $2$-Toda hierarchy,
the information is encoded in a function
$f\colon \bZ +\half \to \bC$ which can be regarded as a substitute of affine coordinates,
and we are able to express the connected $(n,m)$-point functions in terms of $f$.

\subsection{Main results}

Now we state our main results of this paper.
Let
$f\colon \bZ +\half \to \bC$
be an arbitrary function defined on the set of half-integers,
and let
\begin{equation*}
	\label{eq-def-hatf}
	\hf = \sum_{s\in \bZ+\half} f(s) {:}\psi_s \psi_{-s}^*{:}
\end{equation*}
be an operator on the fermionic Fock space,
then
\begin{equation*}
	\tau_f (\bm t^+,\bm t^-) = \big\langle 0\big| \Gamma_+(\bm t^+) \exp\big(\hf\big) \Gamma_-(\bm t^-) \big|0\big\rangle
\end{equation*}
is a diagonal tau-function of the 2d Toda lattice hierarchy.
Our main theorem of the paper is the following formula for the connected $(n,m)$-point functions:
\begin{Theorem}
	The connected $(n,m)$-point functions are given by
	\begin{gather}
		 \sum_{j_1,\dots,j_n,k_1,\dots,k_m \geq 1}
		\frac{\pd^{m+n} \log\tau_f(\bm t^+,\bm t^-)}{\pd t_{j_1}^+ \cdots
			\pd t_{j_n}^+ \pd t_{k_1}^- \cdots \pd t_{k_m}^- }
		\bigg|_{\bm t=0} \cdot
		\prod_{a=1}^n z_a^{-j_a-1} \cdot \prod_{b=1}^m z_{n+b}^{k_b-1} \nonumber\\
		\qquad\qquad = (-1)^{n+m-1} \sum_{\text{$(n+m)$-{\rm cycles}}}
		\prod_{i=1}^{n+m} B_{\sigma(i), \sigma(i+1)}
		- \frac{\delta_{n,2}\delta_{m,0} + \delta_{n,0}\delta_{m,2} }{(z_1-z_2)^2} ,\label{eq-intro-main}
	\end{gather}
	where the summations are taken over all $(n+m)$-cycles $\sigma$,
	and we denote $\sigma(n+m+1) = \sigma(1)$.
	And $B_{i,j}$ are given by
	\begin{equation*}
		B_{i,j} = \begin{cases}
			\displaystyle \sum_{k\geq 0} z_i^{-k-1} z_j^k & \text{if} \ i< j\leq n \ \text{or} \ n<i< j,\\
			\displaystyle\sum_{k\geq 0} {\rm e}^{-f(k+\half)} z_i^{-k-1}z_j^k & \text{if} \ i\leq n<j,\\
			\displaystyle-\sum_{k\geq 0}z_j^{-k-1} z_i^k & \text{if} \ j<i\leq n \ \text{or} \ n<j< i,\\
		\displaystyle	- \sum_{k\geq 0} {\rm e}^{f(-k-\half)} z_j^{-k-1}z_i^k & \text{if} \ j\leq n<i.
		\end{cases}
	\end{equation*}
\end{Theorem}

Then we are able to apply this formula to the concrete computations of
the connected double Hurwitz numbers mentioned above.
Using the results in~\cite{ggn, Ok1, ssz},
one may find that the corresponding functions $f\colon \bZ+\half \to \bC$
for the various double Hurwitz numbers are as follows:
\begin{itemize}\itemsep=0pt
	\item[(1)]
	for the ordinary double Hurwitz numbers: $f^{(2)}(s) = \frac{s^2}{2}$,
	\item[(2)]
	for the double Hurwitz numbers with completed $r$-cycles: $f^{(r)}(s) = \frac{s^r}{r!}$,
	\item[(3)]
	for the mixed double Hurwitz numbers:
	\begin{equation*}
		f^{\rm mix} (s) =\begin{cases}
		\displaystyle	\frac{s^2}{2}u - \log \prod\limits_{j=1}^{-s-\half} (1-jt)
			& \text{if} \ s<0,\\
	\displaystyle		\frac{s^2}{2}u + \log \prod\limits_{j=1}^{s-\half} (1+jt)
			& \text{if} \ s>0.
		\end{cases}
	\end{equation*}
\end{itemize}

Moreover,
the generating series of
the stationary Gromov--Witten invariants of $\bP^1$ relative to two points $0,\infty \in \bP^1$
is also a diagonal tau-function of the 2d Toda lattice hierarchy.
This was established by Okounkov and Pandharipande~\cite{op} using the GW/Hurwitz correspondence,
and in this case the function $f$ is
\begin{equation*}
	f_{\bP^1} (s) =
	\sum_{i\geq 0} x_i \cdot \frac{s^{i+1}}{(i+1)!},
	\qquad \forall s\in \bZ+\half.
\end{equation*}

Furthermore,
we will fix $\bm t^-$ and regard $\tau_f(\bm t^+,\bm t^-)$
as a tau-function of the KP hierarchy with KP-time variables $\bm t^+$.
Then the following result gives an example of
finding relations between different hierarchies from the fermionic point of view
(see Section~\ref{sec-red-KP} for details):
\begin{Theorem}
	The KP-affine coordinates for
	$\tau_f (\bm t^+,\bm t^-)$ $($with fixed $\bm t^-$ and KP-time variables~$\bm t^+)$ are
	\begin{gather}
	\label{eq-KPaffine-diag}
	a_{n,m}^f = (-1)^n \cdot
	s_{(m|n)}(\bm t^-)
	\cdot {\rm e}^{ f(-m-\half) - f(n+\half) },
	\end{gather}
	for every $m,n\geq 0$.
\end{Theorem}
Combining this result with Zhou's original formula for KP tau-functions (see~\cite[Section~5]{zhou1})
will enable one to compute the single Hurwitz numbers.

The rest of this paper is arranged as follows.
In Section~\ref{sec-pre}, we recall some preliminaries of the boson-fermion correspondence.
In Section~\ref{sec-disconn}, we compute the disconnected fermionic and bosonic
$(n,m)$-point functions of a diagonal tau-function.
Then in Section~\ref{sec-conn-nmpt}, we compute the connected bosonic $(n,m)$-point functions
and prove the formula~\eqref{eq-intro-main}
using the results in Section~\ref{sec-disconn}.
In Section~\ref{sec-red-KP},
we fix $\bm t^-$ and compute the KP-affine coordinates of $\tau_f(\bm t^+,\bm t^-)$.
Finally,
we apply~\eqref{eq-intro-main} to the connected double Hurwitz numbers
and the relative stationary GW invariants of $\bP^1$ in Sections~\ref{sec-app-Hurwitz}
and~\ref{sec-app-P1}, respectively.

\section{Preliminaries}
\label{sec-pre}

In this section,
we recall some preliminaries of the boson-fermion correspondence
and the 2d Toda lattice hierarchy.
See, e.g.,~\cite{jm, djm, Ok2, ut} for more details.

\subsection{Free fermions and fermionic Fock space}

In this subsection,
we recall the semi-infinite wedge construction of the fermionic Fock space $\cF$
and the action of free fermions.
See~\cite{djm} and~\cite[Chapter~14]{kac}.

Let $\bm{a}=(a_1,a_2,\dots)$ be a sequence of half-integers $a_i \in \bZ+\half$
satisfying the condition $a_1<a_2<a_3 <\cdots$.
The sequence $\bm{a}$ is said to be admissible if
\begin{equation*}
	\bigg|\bigg(\bZ_{\geq 0}+\half\bigg)-\{a_1,a_2,\dots\}\bigg|<\infty,
	\qquad
	\bigg|\{a_1,a_2,\dots\}-\bigg(\bZ_{\geq 0}+\half\bigg)\bigg|<\infty.
\end{equation*}
For an admissible sequence $\bm{a} =(a_1,a_2,\dots) $,
we denote by $|\bm a\rangle$ the following semi-infinite wedge product
\begin{equation*}
	| \bm a\rangle =
	z^{a_1} \wedge z^{a_2} \wedge z^{a_3} \wedge \cdots,
\end{equation*}
and denote by $\cF$ the infinite-dimensional vector space of all formal (infinite) summations of the form
\begin{equation*}
	\sum_{\bm a: \text{ admissible}} c_{\bm a} |\bm a\rangle,
	\qquad c_{\bm a} \in \bC.
\end{equation*}

The charge of the basis vector $|\bm a\rangle \in \cF$ is defined to be the following integer:
\begin{equation*}
	\text{charge}(|\bm a\rangle)=
	\bigg|\bigg(\bZ_{\geq 0}+\half\bigg)-\{a_1,a_2,\dots\}\bigg|
	-\bigg|\{a_1,a_2,\dots\}-\bigg(\bZ_{\geq 0}+\half\bigg)\bigg|.
\end{equation*}
This gives us a decomposition of the fermionic Fock space
\[
\cF=\bigoplus_{n\in \bZ} \cF^{(n)},
\]
where $\cF^{(n)}$ is spanned by all basis vectors $|\bm a\rangle$ of charge $n$.
We will denote
\[
|n\rangle = z^{n+\half}\wedge z^{n+\frac{3}{2}}
\wedge z^{n+\frac{5}{2}} \wedge \cdots \in \cF^{(n)},
\]
and in particular,
\[
|0\rangle = z^{\half}\wedge z^{\frac{3}{2}}
\wedge z^{\frac{5}{2}} \wedge \cdots \in \cF^{(0)}.
\]
The vector $|0\rangle $ is called the fermionic vacuum vector.
The subspace $\cF^{(0)}$ has a basis labeled by all partitions of integers $\{\mu\}$.
Let $\mu=\{a_1,a_2,\dots\}$ be a partition where
$\mu_1\geq \mu_2\geq \cdots\geq \mu_l >\mu_{l+1}=\mu_{l+2}=\cdots=0$,
and denote
\begin{gather}
\label{eq-cFbasis-mu}
|\mu\rangle =
z^{\frac{1}{2}-\mu_1}\wedge z^{\frac{3}{2}-\mu_2}\wedge
z^{\frac{5}{2}-\mu_3}\wedge \cdots \in\cF^{(0)},
\end{gather}
then $\{|\mu\rangle\}$ form a basis for $\cF^{(0)}$.

Now we recall the action of free fermions $\psi_r$, $\psi_r^*$ (where $r\in \bZ+\half$) on $\cF$.
Let $\psi_r$, $\psi_r^*$ be the following operators on $\cF$:
\begin{gather}
\label{eq-Cli-wedge-1}
\psi_r |\bm a\rangle = z^r \wedge |\bm a\rangle,
\qquad \forall r\in \bZ+\half,
\end{gather}
 and
\begin{gather}
\label{eq-Cli-wedge-2}
\psi_r^* | \bm a\rangle =
\begin{cases}
	(-1)^{k+1} \cdot z^{a_1}\wedge z^{a_2}\wedge \cdots \wedge \widehat{z^{a_k}} \wedge \cdots
	&\text{if} \ a_k = -r \ \text{for some} \ k,\\
	0 &\text{otherwise.}
\end{cases}
\end{gather}
Then one easily checks that the following anti-commutation relations hold:
\begin{gather}
\label{eq-ferm-anticomm}
[\psi_r,\psi_s]_+ = [\psi_r^*,\psi_s^*]_+ =0,
\qquad
[\psi_r,\psi_s^*]_+ = \delta_{r+s,0}\cdot \Id,
\qquad \forall r,s\in \bZ+\half,
\end{gather}
where the bracket is defined by $[a,b]_+ = ab+ba$.
In other words,
\eqref{eq-Cli-wedge-1} and~\eqref{eq-Cli-wedge-2} define an action of the Clifford algebra
on $\cF$.
The operators $\{\psi_r\}$ all have charge $-1$,
and $\{\psi_r^*\}$ all have charge~$1$.
Moreover,
one easily checks that
\begin{gather*}
\psi_r |0\rangle = \psi_r^* |0\rangle =0,
\qquad
\forall r>0,
\end{gather*}
and every basis vector $|\mu\rangle \in \cF^{(0)}$ (where $\mu$ is a partition)
can be obtained by applying operators~$\{\psi_r,\psi_r^*\}_{r<0}$ to the vacuum $|0\rangle$ in the following way:
\begin{gather}
\label{eq-vectormu-psi}
|\mu\rangle=(-1)^{n_1+\cdots+n_k}\cdot
\psi_{-m_1-\half} \psi_{-n_1-\half}^* \cdots
\psi_{-m_k-\half} \psi_{-n_k-\half}^* |0\rangle,
\end{gather}
where $\mu=(m_1,\dots, m_k \mid n_1,\dots,n_k)$ is the Frobenius notation
(see, e.g.,~\cite{mac} for an introduction)
for the partition $\mu$.
The operators $\{\psi_r,\psi_r^*\}_{r<0}$ are called the fermionic creators,
and $\{\psi_r,\psi_r^*\}_{r>0}$ are called the fermionic annihilators.

Furthermore,
one can define an inner product $(\cdot,\cdot)$ on the Fock space $\cF$ by taking
$\{|\bm a\rangle \mid \text{$\bm a$ is} \allowbreak \text{admissible}\}$ to be an orthonormal basis.
Given two admissible sequences $\bm a$ and~$\bm b$,
we denote by $\langle \bm b | \bm a\rangle =(|\bm a\rangle,|\bm b\rangle)$
the inner product of $|\bm a\rangle$ and $|\bm b\rangle$.
Then $\psi_r$ and $\psi_{-r}^*$ are adjoint to each other with respect to this inner product.
Let $A$ be an arbitrary operator (in terms of $\psi_r$, $\psi_s^*$) on $\cF$,
then the inner product of $|\bm a\rangle$ with $A|\bm b \rangle$
will be denoted by $\langle \bm a |A| \bm b \rangle$.
We will also denote by $\langle A \rangle$ the vacuum expectation value of an operator $A$:
\begin{equation*}
	\langle A \rangle = \langle 0 |A| 0\rangle.
\end{equation*}

\subsection{Cut-and-join operator}

In this subsection we recall the cut-and-join operator and its eigenvalues.
The cut-and-join operator plays an important role in the study of Hurwitz numbers,
see, e.g.,~\cite{ gjv, Ok1, zhou2}.

Let $\mu = (\mu_1,\mu_2,\dots,\mu_l)$ be a partition,
and define
\[
\kappa_\mu = \sum_{i=1}^l \mu_i (\mu_i-2i+1).
\]
In particular,
we denote $\kappa_{(\varnothing)} =0$ for the empty partition.
The cut-and-join operator $K^{(2)}$ on the fermionic Fock space is defined to be
\begin{equation}
\label{eq-def-C&Jopr}
K^{(2)} = \sum_{s\in \bZ+\half} \frac{s^2}{2}
{:}\psi_s \psi_{-s}^*{:}.
\end{equation}
Then one has (see~\cite{Ok1})
\begin{equation}
\label{eq-eigen-C&J}
K^{(2)} |\mu \rangle = \frac{\kappa_\mu}{2} |\mu\rangle,
\end{equation}
where $|\mu\rangle$ is the basis vector~\eqref{eq-cFbasis-mu} of $\cF^{(0)}$.

\subsection{Boson-fermion correspondence}

In this subsection,
we recall the bosonic Fock space and boson-fermion correspondence.
See~\cite{djm} for details.

Let $\alpha_n$ be the following operators on $\cF$:
\[
\alpha_n = \sum_{s\in \bZ+\half} {:}\psi_{-s} \psi_{s+n}^*{:},
\qquad n\in \bZ,
\]
where ${:} \psi_{-s} \psi_{s+n}^* {:}$ denotes the normal-ordered product of fermions
defined by
\begin{equation*}
	{:}\phi_{r_1}\phi_{r_2}\cdots \phi_{r_n}{:}
	=(-1)^\sigma \phi_{r_{\sigma(1)}}\phi_{r_{\sigma(2)}}\cdots\phi_{r_{\sigma(3)}},
\end{equation*}
where $\phi_k$ is either $\psi_k$ or $\psi_k^*$,
and $\sigma\in S_n$ is a permutation such that $r_{\sigma(1)}\leq\cdots\leq r_{\sigma(n)}$.
The operator $\alpha_0$ is called the charge operator on $\cF$.
These operators $\{\alpha_n\}_{n\in \bZ}$ satisfy the following commutation relations:
\begin{gather}
\label{eq-comm-boson}
[\alpha_m,\alpha_n]= m\delta_{m+n,0} \cdot \Id,
\end{gather}
i.e.,
they generate a Heisenberg algebra.
The normal-ordered products for the bosons $\{\alpha_n\}_{n\in \bZ}$
are defined by
\begin{equation*}
	{:}\alpha_{n_1}\cdots\alpha_{n_k}{:}=
	\alpha_{n_{\sigma(1)}}\cdots\alpha_{n_{\sigma(k)}},
\end{equation*}
where $\sigma\in S_k$ such that $n_{\sigma(1)}\leq\cdots \leq n_{\sigma(k)}$.
Denote
\begin{gather*}
\psi(\xi)= \sum_{s\in\bZ+\half} \psi_s \xi^{-s-\half},
\qquad
\psi^*(\xi)= \sum_{s\in\bZ+\half} \psi_s^* \xi^{-s-\half},
\end{gather*}
and
\begin{gather}
\label{eq-gen-bos}
\alpha(\xi)= {:}\psi(\xi)\psi^*(\xi){:}=
\sum_{n\in\bZ} \alpha_n \xi^{-n-1},
\end{gather}
then the commutation relation~\eqref{eq-comm-boson} and the anti-commutation relations~\eqref{eq-ferm-anticomm} are
equivalent to the following operator product expansions, respectively:
\begin{gather*}
		\alpha(\xi)\alpha(\eta)= {:}\alpha(\xi)\alpha(\eta){:}+
		\frac{1}{(\xi-\eta)^2},\\
		\psi(\xi)\psi(\eta)= {:}\psi(\xi)\psi(\eta){:},\\
		\psi^*(\xi)\psi^*(\eta)= {:}\psi^*(\xi)\psi^*(\eta){:},\\
		\psi(\xi)\psi^*(\eta)= {:}\psi(\xi)\psi^*(\eta){:}+\frac{1}{\xi-\eta},
\end{gather*}
Moreover,
we have
\begin{equation*}
		\langle \psi(z) \psi^*(w) \rangle = \langle \psi^*(z) \psi(w) \rangle =
		\sum_{k=0}^\infty z^{-k-1}w^k = i_{z,w} \frac{1}{z-w},
\end{equation*}
where the notation $i_{z,w}$ means expanding on $\{|z|>|w|\}$.

The bosonic Fock space $\cB$ is defined by $\cB:=\Lambda\big[\big[w,w^{-1}\big]\big]$,
where $\Lambda$ is the space of symmetric functions in some formal variables
$\bm x=(x_1,x_2,\dots)$,
and $w$ is a formal variable.
The boson-fermion correspondence is a linear isomorphism $\Phi\colon\cF \to \cB$ of vector spaces,
given by (see, e.g.,~\cite[Section~5]{djm}):
\[
\Phi\colon \
|\bm a\rangle \in \cF^{(m)}
\mapsto
w^m\cdot \big\langle m \big|
{\rm e}^{\sum_{n=1}^\infty \frac{p_n}{n} \alpha_n}
\big| \bm a \big\rangle,
\]
where $p_n = p_n(\bm x) \in \Lambda$ ($n\geq 1$) is the Newton symmetric function of degree $n$.
In particular,
by restricting to $\cF^{(0)}$ one obtains an isomorphism
\begin{gather*}
\cF^{(0)}\to \Lambda,
\qquad
|\mu\rangle \mapsto s_\mu =
\big\langle 0 \big| {\rm e}^{\sum_{n=1}^\infty \frac{p_n}{n} \alpha_n} \big| \mu \big\rangle,
\end{gather*}
where $s_\mu = s_\mu( \bm t)$ is the Schur function (see~\cite{mac} for an introduction)
indexed by the partition~$\mu$, and $\bm t = (t_1,t_2, t_3,\dots)$ where $t_n = \frac{p_n}{n}$.
Using the above isomorphism,
one can represent the bosons $\{\alpha_n\}$ and fermions $\{\psi_r,\psi_s^*\}$
as operators on the bosonic Fock space.
One has
\begin{gather*}
\Phi (\alpha_n | \bm a\rangle)
=\begin{cases}
	\displaystyle n\frac{\pd}{\pd p_n} \Phi(|\bm a\rangle), & n>0,\\
	p_{-n}\cdot \Phi(|\bm a\rangle), & n<0,
\end{cases}
\end{gather*}
and
\begin{equation*}
	\Phi(\psi(\xi)|\bm a\rangle)=
	\Psi(\xi) \Phi(|\bm a\rangle),
	\qquad
	\Phi(\psi^*(\xi)|\bm a\rangle)=
	\Psi^*(\xi) \Phi(|\bm a\rangle),
\end{equation*}
where $\Psi(\xi),\Psi^*(\xi)$ are the vertex operators
\begin{gather*}
		\Psi(\xi)=
		\exp\Bigg( \sum_{n=1}^\infty \frac{p_n}{n} \xi^n \Bigg)
		\exp\Bigg( -\sum_{n=1}^\infty \xi^{-n}\frac{\pd}{\pd p_n} \Bigg)
		{\rm e}^K \xi^{\alpha_0},\\
		\Psi^*(\xi)=
		\exp\Bigg( -\sum_{n=1}^\infty \frac{p_n}{n} \xi^n \Bigg)
		\exp\Bigg( \sum_{n=1}^\infty \xi^{-n}\frac{\pd}{\pd p_n} \Bigg)
		{\rm e}^{-K} \xi^{-\alpha_0},	
\end{gather*}
and the actions of ${\rm e}^K$ and $\xi^{\alpha_0}$ are defined by
\begin{equation*}
	\big({\rm e}^K f\big) (z,T)=z\cdot f(z,T),
	\qquad
	\big(\xi^{\alpha_0} f\big)(z,T) = f(\xi z,T).
\end{equation*}

\subsection{Tau-functions of 2d Toda lattice hierarchy}

Now we recall the construction of tau-functions of the 2d Toda lattice hierarchy
as vacuum expectation values.
See~\cite{ta, ut} and~\cite[Appendix]{Ok2}.

Let $\bm t^\pm = \big(t_1^\pm,t_2^\pm,t_3^\pm,\dots\big)$ be two sequences of formal variables.
Denote
\[
\Gamma_\pm \big(\bm t^\pm\big) = \exp \bigg(
\sum_{n\geq 1} t_n^\pm \alpha_{\pm n} \bigg).
\]
Then one has (see, e.g.,~\cite[Appendix]{Ok2}):
\begin{gather}
\label{eq-Y--expansion}
\Gamma_-(\bm t^-) |0\rangle = \sum_{\mu} s_\mu(\bm t^-) |\mu\rangle.
\end{gather}

Now let $A$ be an operator of charge $0$ on $\cF$,
and define
\begin{equation*}
	\tau_n (\bm t^+,\bm t^-) = \big\langle n \big| \Gamma_+(\bm t^+) A \Gamma_-(\bm t^-) \big|n\big\rangle,
	\qquad n\in \bZ,
\end{equation*}
then by the boson-fermion correspondence one has
\begin{gather}
\label{eq-bf-bosonL}
\big\langle n\big| \Gamma_+(\bm t^+) \alpha_k A \Gamma_- (\bm t^-) \big|n\big\rangle
= \begin{cases}
\displaystyle	\frac{\pd}{\pd t_k^+} \tau_n, & k>0,\\
	-kt_{-k}^+ \cdot \tau_n, & k<0.
\end{cases}
\end{gather}
Moreover,
since $\alpha_n$ and $\alpha_{-n}$ are adjoint to each other with respect to
the inner product $(\cdot, \cdot)$ on~$\cF$,
one has
\begin{gather}
\label{eq-bf-bosonR}
\big\langle n\big| \Gamma_+(\bm t^+) A \alpha_k \Gamma_- (\bm t^-) \big|n\big\rangle
= \begin{cases}
	kt_{k}^- \cdot \tau_n, & k>0,\\
\displaystyle	\frac{\pd}{\pd t_{-k}^-} \tau_n, & k<0.
\end{cases}
\end{gather}
The function $\tau_n(\bm t^+,\bm t^-)$ is a tau-function of the 2d Toda lattice hierarchy if
\begin{equation*}
	\bigg[A\otimes A , \sum_{s\in \bZ+\half} \psi_s^* \otimes \psi_{-s}\bigg] =0.
\end{equation*}
This relation is equivalent to the Hirota bilinear relation~\cite[formula (1.3.26)]{ut} of
the 2d Toda lattice hierarchy.

Let $\tau_n(\bm t^+ ,\bm t^-)$ be a tau-function of the 2d Toda lattice hierarchy,
then for every fixed integer~$n$ and time $\bm t^+$,
the function $\tau_n(\bm t^+ ,\bm t^-)$ is a tau-function of the KP hierarchy with KP-time variables~$\bm t^-$;
and similarly for fixed integer $n$ and time $\bm t^-$,
$\tau_n(\bm t^+ ,\bm t^-)$ is a tau-function of the KP hierarchy with KP-time variables $\bm t^+$.

\section[Bosonic (n,m)-point functions of diagonal tau-functions]{Bosonic $\boldsymbol{(n,m)}$-point functions of diagonal tau-functions}
\label{sec-disconn}

In this paper,
we will focus on the so-called diagonal tau-function of the 2d Toda lattice hierarchy,
since many interesting examples of tau-functions coming from algebraic geometry are of this form.
Let
\begin{gather}
\label{eq-def-functionf}
f\colon \ \bZ +\half \to \bC
\end{gather}
be a function defined on the set of half-integers,
and denote by $\hf$ the following element in the infinite-dimensional
Lie algebra $\widehat{\mathfrak{gl} (\infty)}$:
\begin{gather}
\label{eq-def-hatf}
\hf = \sum_{s\in \bZ+\half} f(s) {:}\psi_s \psi_{-s}^*{:}
= \sum_{s<0} f(s) \psi_s \psi_{-s}^* - \sum_{s>0} f(s)\psi_{-s}^*\psi_s.
\end{gather}
A diagonal tau-function is of the following form:
\begin{gather}
\label{eq-def-tauf}
\tau_f = \big\langle 0\big| \Gamma_+(\bm t^+) \exp\big(\hf\big) \Gamma_-(\bm t^-) \big|0\big\rangle.
\end{gather}

In this section,
we first compute the fermionic $(2n,2m)$-point functions associated to $\tau_f$,
and then use this result and the boson-fermion correspondence to
compute the (disconnected) bosonic $(n,m)$-point functions.

\begin{Remark}
	The results for the tau-functions
	\begin{equation*}
		\tau_{n,f} = \big\langle n\big| \Gamma_+(\bm t^+) \exp\big(\hf\big) \Gamma_-(\bm t^-) \big|n\big\rangle
	\end{equation*}
	can be obtained by simply shifting the arguments of the function $f$ by $n$.
\end{Remark}

\subsection{Some basic computations}

Let $\hf$ be the operator~\eqref{eq-def-hatf}.
In this subsection,
we discuss some properties of $\hf$ which will be useful in the following subsections.

First,
it is clear that
\begin{gather}
\label{eq-vacuum-expf}
\exp\big(\hf\big) |0\rangle = |0\rangle,
\qquad
\langle 0| \exp(-\hf) = \langle 0|.
\end{gather}
Moreover,
using the anti-commutation relations~\eqref{eq-ferm-anticomm} we may easily check that
\begin{equation*}
	\big[\hf , \psi_r\big] = f(r)\psi_r,
	\qquad
	\big[\hf , \psi_r^*\big] = -f(-r) \psi_r^*,
	\qquad \forall r\in \bZ+\half.
\end{equation*}
Then by the Baker--Campbell--Hausdorff formula, we have
\begin{gather}
	{\rm e}^{-\hf} \psi_r {\rm e}^\hf =
	\psi_r - \big[\hf, \psi_r\big] + \frac{1}{2!} \big[\hf,\big[\hf,\psi_r\big]\big] -\cdots
	= {\rm e}^{-f(r)} \psi_r, \nonumber\\
	{\rm e}^{-\hf} \psi_r^* {\rm e}^\hf =
	\psi_r - \big[\hf, \psi_r^*\big] + \frac{1}{2!} \big[\hf,\big[\hf,\psi_r^*\big]\big] -\cdots
	= {\rm e}^{f(-r)} \psi_r^*. \label{eq-conj-f-psi}
\end{gather}
We will denote
\begin{gather}
	\psi_f (z) = {\rm e}^{-\hf} \psi(z) {\rm e}^\hf = \sum_r {\rm e}^{-f(r)} \psi_r z^{-r-\half},\nonumber\\
	\psi_f^* (z) = {\rm e}^{-\hf} \psi^*(z) {\rm e}^\hf = \sum_r {\rm e}^{f(-r)} \psi_r^* z^{-r-\half}.
\label{eq-conjhf-psi}
\end{gather}
Denote by $A_f(z,w)$ the following series:
\begin{gather*}
A_f(z,w)= \sum_{k=0}^\infty {\rm e}^{f(k+\half)} z^{-k-1}w^k,
\end{gather*}
then we have
\begin{gather}
	\langle \psi_f (z) \psi^*(w) \rangle
	= \sum_{k=0}^\infty {\rm e}^{-f(k+\half)} z^{-k-1}w^k = A_{-f}(z,w),\nonumber\\
	 \langle \psi_f^* (z) \psi(w) \rangle
	= \sum_{k=0}^\infty {\rm e}^{f(-k-\half)} z^{-k-1}w^k =A_{f(-\cdot)}(z,w),
\label{eq-expvalue-psif-1}
\end{gather}
where the notation $f(-\cdot)$ means the function $r\mapsto f(-r)$ for $r\in \bZ+\half$.
Moreover,
\begin{gather}
\label{eq-expvalue-psif-2}
\langle \psi_f(z) \psi_f^*(w) \rangle
= \langle \psi_f^*(z) \psi_f(w) \rangle
= i_{z,w} \frac{1}{z-w} = A_0(z,w).
\end{gather}

\subsection[Computation of fermionic (2n,2m)-point functions]{Computation of fermionic $\boldsymbol{(2n,2m)}$-point functions}

In this subsection we compute the
fermionic $(2n,2m)$-point functions associated to $\tau_f$,
where $\tau_f$ is the tau-function~\eqref{eq-def-tauf} of the 2d Toda lattice hierarchy.

The fermionic $(2n,2m)$-point function associated to $\tau_f$
is defined to be
\begin{gather}
\label{eq-def-2n2m-ferm}
\big\langle \psi(z_1)\psi^*(w_1) \cdots \psi(z_n)\psi^*(w_n) {\rm e}^\hf
\psi(z_{n+1})\psi^*(w_{n+1}) \cdots \psi(z_{n+m})\psi^*(w_{n+m}) \big\rangle.
\end{gather}
By~\eqref{eq-vacuum-expf}, we can rewrite~\eqref{eq-def-2n2m-ferm} as
\begin{gather}
\label{eq-def-2n2m-ferm-2}
\langle \psi_f(z_1)\psi_f^*(w_1) \cdots \psi_f(z_n)\psi_f^*(w_n)
\psi(z_{n+1})\psi^*(w_{n+1}) \cdots \psi(z_{n+m})\psi^*(w_{n+m}) \rangle.
\end{gather}
For simplicity, we denote
\begin{alignat*}{4}
		&\varphi_i = \psi_f(z_i), \qquad &&
		\varphi_i^* = \psi_f^*(w_i) \qquad &&
		\text{for} \ 1\leq i\leq n, &\\
		&\varphi_j = \psi (z_j), \qquad &&
		\varphi_j^* = \psi^* (w_j) \qquad &&
		\text{for} \ n+1\leq j\leq n+m. &
\end{alignat*}
Notice that $\psi_f(z)$, $\psi_f^*(w)$ are of the form~\eqref{eq-conjhf-psi},
thus we may apply Wick's theorem (see, e.g.,~\cite[Section~4.5]{djm})
to~\eqref{eq-def-2n2m-ferm-2}.
Since $\langle \varphi_i\varphi_j \rangle = \langle \varphi_i^*\varphi_j^* \rangle =0$,
we rewrite~\eqref{eq-def-2n2m-ferm-2} as
\begin{gather}
\label{eq-fermnpt-wicksum}
\langle \varphi_1 \varphi_1^* \cdots \varphi_{n+m}\varphi_{n+m}^* \rangle
= \sum_{\bm p} \operatorname{sgn}(\bm p)\cdot C(\bm p)_1 C(\bm p)_2 \cdots C(\bm p)_{n+m},
\end{gather}
where the sequence $\bm p = (p_1,\dots,p_{n+m})$ runs over permutations of $(1,2,\dots,n+m)$,
and $\operatorname{sgn}(\bm p) = \pm 1$ is the sign of this permutation,
and $C(\bm p)_i$ is given by
\[
C(\bm p)_i = \begin{cases}
	\langle \varphi_i \varphi_{p_i}^* \rangle & \text{if} \ p_i\geq i,\\
	- \langle \varphi_{p_i}^* \varphi_i \rangle & \text{if} \ p_i< i.
\end{cases}
\]
Then by plugging~\eqref{eq-expvalue-psif-1} and~\eqref{eq-expvalue-psif-2} into
the above definition of $C(\bm p)_i$,
we may easily see that the right-hand side of~\eqref{eq-fermnpt-wicksum} is a determinant.
\begin{Theorem}
	The fermionic $(2n,2m)$-point function is given by
	\begin{equation*}
			\big\langle \psi(z_1)\psi^*(w_1) \cdots \psi(z_n)\psi^*(w_n) {\rm e}^\hf
			\psi(z_{n+1})\psi^*(w_{n+1}) \cdots \psi(z_{n+m})\psi^*(w_{n+m}) \big\rangle
			= \det(C_{i,j}),
	\end{equation*}
	where $(C_{i,j})$ is the following $(n+m)\times(n+m)$ matrix:
	\begin{equation*}
		C_{i,j} = \begin{cases}
			\langle \varphi_i\varphi_j^* \rangle, & i\leq j,\\
			-\langle \varphi_j^* \varphi_i \rangle, & i>j,
		\end{cases}
	\end{equation*}
	or more precisely,
	\begin{equation*}
		C_{i,j} = \begin{cases}
			A_0 (z_i,w_j) & \text{if} \ i\leq j\leq n \ \text{or} \ n<i\leq j,\\
			A_{-f} (z_i,w_j) & \text{if} \ i\leq n<j,\\
			-A_0(w_j,z_i) & \text{if} \ j<i\leq n \ \text{or} \ n<j< i,\\
			-A_{f(-\cdot)}(w_j,z_i) & \text{if} \ j\leq n<i.
		\end{cases}
	\end{equation*}
\end{Theorem}

\subsection[Computation of bosonic (n,m)-point functions]{Computation of bosonic $\boldsymbol{(n,m)}$-point functions}

In this section, we compute the following bosonic $(n,m)$-point function
associated to the tau-function $\tau_f$:
\begin{gather}
\label{eq-def-2n2m-bos}
\big\langle \alpha(z_1)\cdots \alpha(z_n) {\rm e}^\hf
\alpha(z_{n+1}) \cdots \alpha(z_{n+m}) \big\rangle,
\end{gather}
where $\alpha(z)$ is the generating series~\eqref{eq-gen-bos} of bosons.

First, we need to compute the normally ordered fermionic $(2n,2m)$-point function
using the main result in last subsection.
We have
\begin{Proposition}
	The normally ordered fermionic $(2n,2m)$-point function
	\begin{align}
		\big\langle& {:}\psi(z_1)\psi^*(w_1){:} {:}\psi(z_2)\psi^*(w_2){:} \cdots {:}\psi(z_n)\psi^*(w_n){:} {\rm e}^\hf \nonumber\\
		& {:}\psi(z_{n+1})\psi^*(w_{n+1}){:} \cdots {:}\psi(z_{n+m})\psi^*(w_{n+m}){:} \big\rangle \label{eq-ferm2n2mpt-norm}
	\end{align}
	equals to the determinant $\det\big(\tC_{i,j}\big)$,
	where $\big(\tC_{i,j}\big)$ is the $(n+m)\times(n+m)$ matrix
	\[
	\tC_{i,j} = \begin{cases}
		0 & \text{if} \ i=j,\\
		C_{i,j} & \text{if} \ i \not= j.
	\end{cases}
	\]
\end{Proposition}
\begin{proof}
	Recall that we have
	\begin{equation*}
		{:}\psi(z)\psi^*(w){:} = \psi(z)\psi^*(w) - \frac{1}{z-w}
	\end{equation*}
	and $\langle 0|=\langle 0| {\rm e}^{-\hf}$,
	thus~\eqref{eq-ferm2n2mpt-norm} equals to
	\begin{gather}
		 \bigg\langle
		\bigg( \psi(z_1)\psi^*(w_1) - \frac{1}{z_1-w_1} \bigg) \cdots
		\bigg( \psi(z_n)\psi^*(w_n) - \frac{1}{z_n-w_n} \bigg) {\rm e}^\hf
\nonumber\\
		 \qquad\quad\times \bigg( \psi(z_{n+1})\psi^*(w_{n+1}) - \frac{1}{z_{n+1}-w_{n+1}} \bigg) \cdots
		\bigg( \psi(z_{n+m})\psi^*(w_{n+m}) - \frac{1}{z_{n+m}-w_{n+m}}\bigg) \bigg\rangle
\nonumber\\
		\qquad=
		\bigg\langle
		\bigg( \psi_f(z_1)\psi_f^*(w_1) - \frac{1}{z_1-w_1} \bigg) \cdots
		\bigg( \psi_f(z_n)\psi_f^*(w_n) - \frac{1}{z_n-w_n} \bigg)
\nonumber\\
		 \qquad\quad\times \bigg( \psi(z_{n+1})\psi^*(w_{n+1}) - \frac{1}{z_{n+1}-w_{n+1}} \bigg) \cdots
		\bigg( \psi(z_{n+m})\psi^*(w_{n+m}) - \frac{1}{z_{n+m}-w_{n+m}} \bigg) \bigg\rangle
\nonumber\\
		\qquad= \sum_{K_1\sqcup L_1 = [n], K_2\sqcup L_2 =[m]}
		\bigg( \prod_{l\in L_1} f_l \bigg) \bigg( \prod_{l\in L_2} f_l' \bigg)
		\langle \psi_{K_1}\psi_{K_2}' \rangle, \label{eq-disconn-norm-pf}
	\end{gather}
	where $[n]$ denotes the set $\{1,2,\dots,n\}$, and
	\begin{equation*}
		f_l = -\frac{1}{z_l - w_l} = -A_0(z_l,w_l),
		\qquad
		f_l' = -\frac{1}{z_{n+l}-w_{n+l}} =-A_0(z_{n+l},w_{n+l}),
	\end{equation*}
	and for a set of indices $K=\{k_1,k_2,\dots,k_s\}$ with $k_1<k_2<\cdots <k_s$,
	\begin{gather*}
			\psi_K = \psi_f(z_{k_1})\psi_f^*(w_{k_1})\psi_f(z_{k_2})\psi_f^*(w_{k_2})\cdots \psi_f(z_{k_s})\psi_f^*(w_{k_s}),\\
			\psi_K' = \psi(z_{n+k_1})\psi^*(w_{n+k_1})\psi(z_{n+k_2})\psi^*(w_{n+k_2})\cdots \psi(z_{n+k_s})\psi^*(w_{n+k_s}).
	\end{gather*}
	Now we apply Wick's theorem to $\langle\psi_{K_1}\psi_{K_2}'\rangle$,
	and then compare the resulting summation with formula~\eqref{eq-fermnpt-wicksum}.
	In this way we easily see that~\eqref{eq-disconn-norm-pf} equals to
	determinant $\det (\tC_{i,j})$,
	where the $(n+m)\times (n+m)$ matrix $\big(\tC_{i,j}\big)$ should be
	\begin{equation*}
		\tC_{i,j} =
		C_{i,j} - \delta_{i,j}\cdot\frac{1}{z_i - w_i},
		\qquad 1\leq i\leq n+m.
	\end{equation*}
	Then the conclusion holds because $C_{i,i} = A_0 (z_i,w_i) = \frac{1}{z_i-w_i}$.
\end{proof}

Now recall that $\alpha(z) = {:}\psi(z)\psi^*(z){:}$.
Thus by taking $w_i \to z_i$ for every $i$ in the above proposition,
we obtain
\begin{Theorem}
	\label{thm-disconn-bos-det}
	The bosonic $(n,m)$-point function is given by
	\[
	\langle \alpha(z_1)\cdots \alpha(z_n) {\rm e}^\hf
	\alpha(z_{n+1}) \cdots \alpha(z_{n+m}) \rangle
	= \det (B_{i,j}),
	\]
	where $(B_{i,j})$ is the following $(n+m)\times (n+m)$ matrix:
	\begin{gather}
	\label{eq-matrixcoeff-bij}
	B_{i,j} = \begin{cases}
		A_0 (z_i,z_j) & \text{if} \ i< j\leq n \ \text{or} \ n<i< j,\\
		A_{-f} (z_i,z_j) & \text{if} \ i\leq n<j,\\
		-A_0(z_j,z_i) & \text{if} \ j<i\leq n \ \text{or} \ n<j< i,\\
		-A_{f(-\cdot)}(z_j,z_i) & \text{if} \ j\leq n<i,\\
		0 & \text{if} \ i=j.
	\end{cases}
	\end{gather}
\end{Theorem}

\begin{Remark}Notice that the notation $B_{i,j}$ depends on the choice of $(n,m)$.
\end{Remark}

\begin{Remark}
	\label{rmk-disconn-1&2}
	In particular,
	for $(n,m) = (1,0)$ or $(0,1)$, one easily sees that
	\begin{gather}
	\label{eq-disconn-1001=0}
	\big\langle \alpha(z_1) {\rm e}^\hf \big\rangle =
	\big\langle {\rm e}^\hf \alpha(z_1) \big\rangle = 0.
	\end{gather}
	And for $(n,m) = (2,0)$ or $(0,2)$,
	we have
	\begin{gather*}
	\big\langle \alpha(z_1)\alpha(z_2) {\rm e}^\hf \big\rangle =
	\big\langle {\rm e}^\hf \alpha(z_1)\alpha(z_2) \big\rangle = i_{z_1,z_2} \frac{1}{(z_1-z_2)^2},
	\end{gather*}
	where
	\begin{equation*}
		i_{z_1,z_2}\frac{1}{(z_1-z_2)^2} = \sum_{n=1}^\infty nz_1^{-n-1} z_2^{n-1}.
	\end{equation*}
\end{Remark}

\section[Computation of connected bosonic (n,m)-point functions]{Computation of connected bosonic $\boldsymbol{(n,m)}$-point functions}
\label{sec-conn-nmpt}

In this subsection, we compute the connected bosonic $(n,m)$-point functions
associated to the tau-function $\tau_f$ of the form~\eqref{eq-def-tauf}.
We also show that the connected bosonic $(n,m)$-point functions are the
generating series of coefficients of the free energy $\log \tau_f (\bm t^+,\bm t^-)$.

\subsection[Connected bosonic (n,m)-point functions]{Connected bosonic $\boldsymbol{(n,m)}$-point functions}

In this subsection,
we introduce the notion of connected bosonic $(n,m)$-point functions.

Let $\tau_f$ be the tau-function defined by~\eqref{eq-def-tauf}.
The connected bosonic $(n,m)$-point functions
$\langle \alpha(z_1) \cdots \alpha(z_n) \alpha(z_{n+1}) \cdots \alpha(z_{n+m}) \rangle_{f;n,m}^c$
associated to $\tau_f$ are defined by the following M\"obius inversion formulas:
\begin{gather*}
		\langle \alpha(z_{[n+m]}) \rangle_{f;m,n}
		= \sum_{I_1\sqcup \cdots \sqcup I_k =[n+m]} \frac{1}{k!}
		\langle \alpha(z_{I_1}) \rangle_{f;n_1,m_1}^c \cdots \langle \alpha(z_{I_k}) \rangle_{f;n_k,m_k}^c,\\
		\langle \alpha(z_{[n+m]}) \rangle_{f;m,n}^c
		=\sum_{I_1\sqcup \cdots \sqcup I_k =[n+m]} \frac{(-1)^{k-1}}{k}
		\langle \alpha(z_{I_1}) \rangle_{f;n_1,m_1} \cdots \langle \alpha(z_{I_k}) \rangle_{f;n_k,m_k},
\end{gather*}
where $[n+m]=\{1,2,\dots,n+m\}$,
and $I_1,\dots,I_k \subset [n+m]$ are nonempty subsets.
On the right-hand side, we denote
\begin{gather}
\label{eq-notation-njmj}
n_j = \big|I_j\cap [n]\big|,
\qquad
m_j = \big|I_j \backslash [n]\big|,
\qquad
1\leq j\leq k.
\end{gather}
Here $\langle \alpha(z_{I_j}) \rangle_{f;n_j,m_j}$ denotes the disconnected
bosonic $(n_j,m_j)$-point function
\begin{equation*}
	\langle \alpha(z_{I_j}) \rangle_{f;n_j,m_j}
	=\big\langle \alpha(z_{i_1}) \cdots \alpha(z_{i_{n_j}}) {\rm e}^\hf
	\alpha\big(z_{i_{n_j+1}}\big) \cdots \alpha (z_{i_{n_j+m_j}} ) \big\rangle,
\end{equation*}
where $I_j = \{i_1,\dots, i_{n_j+m_j}\}$ with $i_1< \cdots < i_{n_j}
\leq n < i_{n_j+1} <\cdots < i_{n_j+m_j}$.
In particular,
$\langle \alpha(z_{[n+m]})\rangle_{f;n,m}$ is exactly the
bosonic $(n,m)$-point function~\eqref{eq-def-2n2m-bos}.

\begin{Remark}
	The above definition of connected $(n,m)$-point functions via M\"obius inversion formulas
	is motivated by the inclusion-exclusion principle,
	see Rota~\cite{ro}.
	In the case of double Hurwitz numbers,
	the disconnected $(n,m)$-point functions count disconnected covers between Riemann surfaces
	while connected $(n,m)$-point functions count connected covers,
	see Section~\ref{sec-app-Hurwitz} for details.
\end{Remark}

\subsection[Examples for small (n,m)]{Examples for small $\boldsymbol{(n,m)}$}

In this subsection, we compute the connected bosonic $(n,m)$-point functions
for small $(n,m)$.
We represent the results in terms of the free energy
\[
F_f (\bm t^+,\bm t^-)= \log \tau_f (\bm t^+,\bm t^-).
\]

Given a pair of nonnegative integers $(n,m)$ with $(n,m)\not=(0,0)$,
we denote
\begin{gather}
	G_{f;n,m}(z_1,\dots,z_{n+m})\nonumber\\
	\qquad{}= \big\langle \Gamma_+ (\bm t^+)\alpha(z_1)\cdots \alpha(z_n) {\rm e}^\hf
	\alpha(z_{n+1})\cdots \alpha(z_{n+m}) \Gamma_-(\bm t^-)\big\rangle
	/\tau_f (\bm t^+,\bm t^- ),\label{eq-def-Gfnm}
\end{gather}
and denote by $G_{f;n,m}^c$ its connected version (obtained by M\"obius inversion)
\begin{equation*}
	G_{f;n,m}^c (z_1,\dots,z_{n+m})
	=\sum_{I_1\sqcup \cdots \sqcup I_k =[n+m]} \frac{(-1)^{k-1}}{k}
	G_{f;n_1,m_1}(z_{I_1})\cdots G_{f;n_k,m_k}(z_{I_k}),
\end{equation*}
where we use the notation~\eqref{eq-notation-njmj}.
Then the bosonic $(n,m)$-point functions can be obtained by taking $\bm t^+ = \bm t^- =0$:
\begin{gather*}
	\big\langle\alpha\big(z_{[n+m]}\big) \big\rangle_{f;m,n} =
	G_{f;n,m}(z_1,\dots,z_{n+m})|_{\bm t =0},\\
	\big\langle \alpha\big(z_{[n+m]}\big) \big\rangle_{f;m,n}^c =
	G_{f;n,m}^c(z_1,\dots,z_{n+m})|_{\bm t =0}.
\end{gather*}
Here for simplicity, we write $\bm t =0$ instead of $\bm t^+ = \bm t^- =0$.
In the rest of this subsection, we compute some examples of $G_{f;n,m}^c$ for small $(n,m)$.

\begin{Example}
	For $(n,m) = (1,0),(1,0)$,
	by~\eqref{eq-bf-bosonL} and~\eqref{eq-bf-bosonR} we have
	\begin{gather*}
			 G_{f;1,0}(z) = \frac{1}{\tau_f}
			\sum_{k=1}^\infty \bigg(
			z^{-k-1} \frac{\pd}{\pd t_k^+} + kt_k^+ z^{k-1}
			\bigg) \tau_f
			= \sum_{k=1}^\infty \bigg(
			\frac{\pd \log \tau_f}{\pd t_k^+}z^{-k-1} + kt_k^+ z^{k-1}
			\bigg) ,\\
			G_{f;0,1}(z) = \frac{1}{\tau_f}
			\sum_{k=1}^\infty \bigg(
			z^{k-1} \frac{\pd}{\pd t_k^-} + kt_k^- z^{-k-1}
			\bigg) \tau_f
			= \sum_{k=1}^\infty \bigg(
			\frac{\pd \log \tau_f}{\pd t_k^-}z^{k-1} + kt_k^- z^{-k-1}
			\bigg) .
		\end{gather*}
	By M\"obius inversion formulas, we know
	\begin{gather}
	\label{eq-Gc1001}
	G_{f;1,0}^c (z) = G_{f;1,0} (z) ,
	\qquad
	G_{f;0,1}^c (z) = G_{f;0,1} (z) ,
	\end{gather}
	and then by restricting to $\bm t^+ = \bm t^- = 0$, we get
	\begin{gather*}
			\langle \alpha(z) \rangle_{f;1,0}^c = \sum_{n=1}^\infty
			\frac{\pd F_f(\bm t^+ ,\bm t^-)}{\pd t_n^+} \bigg|_{\bm t = 0}
			\cdot z^{-n-1},\\
			\langle \alpha(z) \rangle_{f;0,1}^c = \sum_{n=1}^\infty
			\frac{\pd F_f(\bm t^+ ,\bm t^-)}{\pd t_n^-} \bigg|_{\bm t = 0}
			\cdot z^{n-1}.
		\end{gather*}
	This is equivalent to say that the connected bosonic $1$-point functions are generating series
	of the coefficients of linear terms in the free energy $F_f (\bm t^+,\bm t^-)$.
	Moreover,
	by~\eqref{eq-disconn-1001=0} we already know
	$G_{f;1,0} (z) |_{\bm t =0} = G_{f;1,0} (z) |_{\bm t =0} =0$,
	thus we obtain
	\begin{gather*}
	\frac{\pd F_f(\bm t^+ ,\bm t^-)}{\pd t_n^+} \bigg|_{\bm t = 0}=0,
	\qquad
	\frac{\pd F_f(\bm t^+ ,\bm t^-)}{\pd t_n^-} \bigg|_{\bm t = 0}=0,
	\qquad \forall n\geq 1.
	\end{gather*}

\end{Example}

\begin{Example}
	For $(n,m) = (1,1)$,
	we have
	\begin{gather*}
			G_{f;1,1}(z_1,z_2) = \frac{1}{\tau_f}
			\sum_{k,l=1}^\infty \bigg(
			z_1^{-k-1} \frac{\pd}{\pd t_k^+} + kt_k^+ z_1^{k-1}
			\bigg)
			\bigg(
			z_2^{l-1} \frac{\pd}{\pd t_l^-} + lt_l^- z_2^{-l-1}
			\bigg) \tau_f \\
\hphantom{G_{f;1,1}(z_1,z_2)}{}
			= \sum_{k,l=1}^\infty \bigg(
			\frac{\pd^2 F_f}{\pd t_k^+ \pd t_l^-} z_1^{-k-1} z_2^{l-1} +
			\frac{\pd F_f}{\pd t_k^+ }
			\frac{\pd F_f}{ \pd t_l^-} z_1^{-k-1} z_2^{l-1}+ lt_l^- \frac{\pd F_f}{\pd t_k^+ } z_1^{-k-1}z_2^{-l-1} \\
\hphantom{G_{f;1,1}(z_1,z_2)= \sum_{k,l=1}^\infty \bigg(}{}
			+ kt_k^+ \frac{\pd F_f}{\pd t_l^- } z_1^{k-1}z_2^{l-1}
			+ klt_k^+ t_l^- z_i^{k-1}z_2^{-l-1} \bigg).
		\end{gather*}
	And by M\"obius inversion,
	\begin{align}
		G_{f;1,1}^c(z_1,z_2)
		= G_{f;1,1}(z_1,z_2) - G_{f;1,0}(z_1) G_{f;0,1}(z_2)
		= \sum_{k,l=1}^\infty
		\frac{\pd^2 F_f}{\pd t_k^+ \pd t_l^-} \cdot z_1^{-k-1} z_2^{l-1},	\label{eq-Gc11}
	\end{align}
	thus by restricting to $\bm t^+ = \bm t^- =0$, we get
	\begin{equation*}
		\langle \alpha(z_1)\alpha(z_2) \rangle_{f;1,1}^c =
		\sum_{k,l=1}^\infty
		\frac{\pd^2 F_f(\bm t^+, \bm t^-)}{\pd t_k^+ \pd t_l^-}\bigg|_{\bm t=0} \cdot z_1^{-k-1} z_2^{l-1}.
	\end{equation*}
\end{Example}

\begin{Example}
	Now consider $(n,m) = (2,0)$.
We have
	\begin{align*}
			G_{f;2,0}(z_1,z_2) ={}& \frac{1}{\tau_f}
			\sum_{k,l=1}^\infty \bigg(
			z_1^{-k-1} \frac{\pd}{\pd t_k^+} + kt_k^+ z_1^{k-1}
			\bigg)
			\bigg(
			z_2^{-l-1} \frac{\pd}{\pd t_l^+} + lt_l^+ z_2^{l-1}
			\big) \tau_f\\
			={}& \sum_{k=1}^\infty kz_1^{-k-1} z_2^{k-1} +
			\sum_{k,l=1}^\infty \bigg( \bigg(
			\frac{\pd^2 F_f}{\pd t_k^+ \pd t_l^+} +
			\frac{\pd F_f}{\pd t_k^+ }
			\frac{\pd F_f}{ \pd t_l^+} \bigg) z_1^{-k-1} z_2^{-l-1} \\
			&{} +lt_l^+ \frac{\pd F_f}{\pd t_k^+ } z_1^{-k-1}z_2^{l-1}
			+ kt_k^+ \frac{\pd F_f}{\pd t_l^+ } z_1^{k-1}z_2^{-l-1}
			+ klt_k^+ t_l^+ z_i^{k-1}z_2^{l-1} \bigg),
		\end{align*}
	and
	\begin{align}
		G_{f;2,0}^c(z_1,z_2)
		&= G_{f;2,0}(z_1,z_2) - G_{f;1,0}(z_1) G_{f;1,0}(z_2)\nonumber\\
		&= i_{z_1,z_2} \frac{1}{(z_1-z_2)^2} +
		\sum_{k,l=1}^\infty
		\frac{\pd^2 F_f}{\pd t_k^+ \pd t_l^+} \cdot z_1^{-k-1} z_2^{-l-1},\label{eq-Gc20}
	\end{align}
	where
	\begin{equation*}
		i_{z_1,z_2}\frac{1}{(z_1-z_2)^2} = \sum_{n=1}^\infty nz_1^{-n-1} z_2^{n-1}.
	\end{equation*}
	Thus, we have
	\begin{equation*}
		\langle \alpha(z_1)\alpha(z_2) \rangle_{f;2,0}^c =
		i_{z_1,z_2}\frac{1}{(z_1-z_2)^2} +
		\sum_{k,l=1}^\infty
		\frac{\pd^2 F_f(\bm t^+,\bm t^-)}{\pd t_k^+ \pd t_l^+}\bigg|_{\bm t=0} \cdot z_1^{-k-1} z_2^{-l-1}.
	\end{equation*}
	Similarly,
	one can also check that for $(n,m)=(0,2)$, we have
	\begin{equation}
	\label{eq-Gc02}
	G_{f;0,2}^c (z_1,z_2) =
	i_{z_1,z_2}\frac{1}{(z_1-z_2)^2} +
	\sum_{k,l=1}^\infty
	\frac{\pd^2 F_f(\bm t^+,\bm t^-)}{\pd t_k^- \pd t_l^-} \cdot z_1^{k-1} z_2^{l-1},
	\end{equation}
	and thus
	\begin{equation*}
		\langle \alpha(z_1)\alpha(z_2) \rangle_{f;0,2}^c =
		i_{z_1,z_2}\frac{1}{(z_1-z_2)^2} +
		\sum_{k,l=1}^\infty
		\frac{\pd^2 F_f(\bm t^+,\bm t^-)}{\pd t_k^- \pd t_l^-}\bigg|_{\bm t=0} \cdot z_1^{k-1} z_2^{l-1}.
	\end{equation*}
	And again by Remark~\ref{rmk-disconn-1&2},
	we finally see
	\begin{equation*}
	\frac{\pd^2 F_f(\bm t^+,\bm t^-)}{\pd t_k^+ \pd t_l^+}\bigg|_{\bm t=0} =0,
	\qquad
	\frac{\pd^2 F_f(\bm t^+,\bm t^-)}{\pd t_k^- \pd t_l^-}\bigg|_{\bm t=0} =0,
	\qquad \forall k,l\geq 1.
	\end{equation*}	
\end{Example}

\begin{Example}
	For $(n,m)=(2,1)$,
	concrete computation tells us (here we omit the details)
	\begin{equation*}
		G_{f;2,1}^c (z_1,z_2,z_3) = \sum_{j,k,l\geq 1}
		\frac{\pd^3 F_f(\bm t^+,\bm t^-)}{ \pd t_j^+ \pd t_k^+ \pd t_l^-}
		\cdot z_1^{-j-1} z_2^{-k-1} z_3^{l-1}.
	\end{equation*}
	Similarly, the result in the case $(n,m)=(1,2)$ is
	\begin{equation*}
		G_{f;1,2}^c (z_1,z_2,z_3)= \sum_{j,k,l\geq 1}
		\frac{\pd^3 F_f(\bm t^+,\bm t^-)}{ \pd t_j^+ \pd t_k^- \pd t_l^-}
		\cdot z_1^{-j-1} z_2^{k-1} z_3^{l-1}.
	\end{equation*}
	For $(n,m)=(0,3)$ and $(3,0)$,
	the results are (here we omit the details)
	\begin{gather*}
			G_{f;3,0}^c (z_1,z_2,z_3) = \sum_{j,k,l\geq 1}
			\frac{\pd^3 F_f(\bm t^+,\bm t^-)}{ \pd t_j^+ \pd t_k^+ \pd t_l^+}
			\cdot z_1^{-j-1} z_2^{-k-1} z_3^{-l-1},\\
			G_{f;0,3}^c (z_1,z_2,z_3) = \sum_{j,k,l\geq 1}
			\frac{\pd^3 F_f(\bm t^+,\bm t^-)}{ \pd t_j^- \pd t_k^- \pd t_l^-}
			\cdot z_1^{j-1} z_2^{k-1} z_3^{l-1}.
		\end{gather*}
	And for $(n,m)=(2,2)$,
	the result is (here we omit the details)
	\begin{equation*}
		G_{f;2,2}^c (z_1,z_2,z_3,z_4) = \sum_{i,j,k,l\geq 0}
		\frac{\pd^4 F_f(\bm t^+,\bm t^-)}{ \pd t_i^+ \pd t_j^+ \pd t_k^- \pd t_l^-}
		\cdot z_1^{-i-1} z_2^{-j-1} z_3^{k-1} z_4^{l-1}.
	\end{equation*}
\end{Example}

\subsection{Relation to the free energy}

Now we present the general relation between the connected bosonic $(n,m)$-point functions
$\big\langle \alpha\big(z_{[n+m]}\big) \big\rangle_{f;m,n}^c$ and the free energy
$F_f = \log \tau_f $
associated to the tau-function $\tau_f$.
First we show that
\begin{Proposition}
	For every $(n,m)$ with $n+m\geq 3$,
	we have
	\begin{gather}	
		G_{f;n,m}^c (z_1,\dots,z_{n+m})\nonumber\\
		\qquad {}= \sum_{j_1,\dots,j_n,k_1,\dots,k_m \geq 1}
		\frac{\pd^{m+n} F_f(\bm t^+,\bm t^-)}{\pd t_{j_1}^+ \cdots \pd t_{j_n}^+ \pd t_{k_1}^- \cdots \pd t_{k_m}^- }
		\cdot
		\prod_{a=1}^n z_a^{-j_a-1} \cdot \prod_{b=1}^m z_{n+b}^{k_b-1}.\label{eq-prop-mobF}
	\end{gather}
\end{Proposition}
\begin{proof}
	This proposition is actually a straightforward modification of Zhou~\cite[Proposition~5.1]{zhou1}.
	Here we use the same method to prove this identity.
	
	We prove this by induction on $n+m$.
	The conclusion holds in all the cases with $n+m =3$ by the examples in last subsection.
	Now assume the conclusion holds for all $(n,m)$ with $n+m \leq i$ (where $i\geq 3$),
	and consider the case with $n+m = i+1$.
	Denote
	\begin{equation}	
 \label{eq-def-tG-leq2}
	\tG_{f;n,m}^c (z_1,\dots,z_{n+m}) = G^c_{f;n,m} (z_1,\dots,z_{n+m}),
	\qquad \text{for} \ n+m\leq 2,
	\end{equation}	
 (see~\eqref{eq-Gc1001},~\eqref{eq-Gc11},~\eqref{eq-Gc20}, and~\eqref{eq-Gc02});
	and for $n+m\geq 3$,
	denote by $\tG_{f;n,m}^c (z_1,\dots,z_{n+m})$ the right-hand side of~\eqref{eq-prop-mobF}.
	By the induction hypothesis,
	we have
	\[
 \tG_{f;n,m}^c (z_1,\dots,z_{n+m}) = G_{f;n,m}^c (z_1,\dots,z_{n+m}),
	\qquad \text{for} \ n+m\leq i.
	\]	
	Now fix a pair $(n,m)$ with $n+m = i+1$.
	Without loss of generality, we may assume $n>0$
	(otherwise $m>0$ and similar arguments will work).
	Denote
	\[	
 \tilde\alpha(z)=
	\sum_{k=1}^\infty \bigg(
	z^{-k-1} \frac{\pd}{\pd t_k^+} + kt_k^+ z^{k-1}
	\bigg),
	\]	
 then by the definition~\eqref{eq-def-Gfnm} and the induction hypothesis,
	we have
	\begin{gather*}
			 G_{f;n,m} (z_1,\dots, z_{n+m})\\
			\qquad{}=
			{\rm e}^{-F_f } \tilde\alpha (z_1) \big( G_{f;n,m} (z_2,\dots, z_{n+m}) {\rm e}^{F_f} \big)\\
			\qquad{}= {\rm e}^{-F_f } \tilde\alpha (z_1) \bigg( {\rm e}^{F_f}
			\sum_{I_1\sqcup \cdots \sqcup I_k = [n+m] \backslash\{1\} } \frac{1}{k!}
			\tG_{f;n_1,m_1}^c(z_{I_i}) \cdots \tG_{f;n_k,m_k}^c(z_{I_k}) \bigg)\\
			\qquad{}= {\rm e}^{-F_f} \big( \tilde\alpha(z_1) {\rm e}^{F_f} \big) \cdot
			\sum_{I_1\sqcup \cdots \sqcup I_k = [n+m]\backslash\{1\} } \frac{1}{k!}
			\tG_{f;n_1,m_1}^c(z_{I_i}) \cdots \tG_{f;n_k,m_k}^c(z_{I_k})\\
			\quad\qquad{}+ \tilde\alpha(z_1)_+
			\sum_{I_1\sqcup \cdots \sqcup I_k = [n+m]\backslash\{1\} } \frac{1}{k!}
			\tG_{f;n_1,m_1}^c(z_{I_1}) \cdots \tG_{f;n_k,m_k}^c(z_{I_k}),
		\end{gather*}
	where
	\begin{equation*}
		\tilde\alpha(z_1)_+ =
		\sum_{k=1}^\infty
		z^{-k-1} \frac{\pd}{\pd t_k^+}.
	\end{equation*}
	Notice that by~\eqref{eq-Gc1001} and~\eqref{eq-def-tG-leq2}, we have
	\begin{equation*}
		{\rm e}^{-F_f} \alpha(z_1) {\rm e}^{F_f} = \tG_{f;1,0}^c (z_1),
	\end{equation*}
	thus the above computation gives
	\begin{gather*}
			 G_{f;n,m} (z_1,\dots, z_{n+m})\\
			\qquad{}=
			\sum_{I_1\sqcup \cdots \sqcup I_k = [n+m]\backslash\{1\} } \Bigg(
			\tG_{f;1,0}^c (z_1)
			\cdot \frac{1}{k!}
			\tG_{f;n_1,m_1}^c(z_{I_i}) \cdots \tG_{f;n_k,m_k}^c(z_{I_k})\\
			\quad\qquad{} + \frac{1}{k!} \sum_{l=1}^k
			\tG_{f;n_1,m_1}^c(z_{I_1}) \cdots
			\tG_{f;n_l+1 ,m_l}^c(z_1,z_{I_l})
			\cdots \tG_{f;n_k,m_k}^c(z_{I_k}) \Bigg).
		\end{gather*}
	Now comparing this with the M\"obius inversion formula
	\begin{equation*}
		G_{f;n,m} (z_1,\dots,z_{n+m})
		=\sum_{I_1\sqcup \cdots \sqcup I_k =[n+m]} \frac{1}{k!}
		G_{f;n_1,m_1}^c(z_{I_1})\cdots G_{f;n_k,m_k}^c(z_{I_k}),
	\end{equation*}
	and using the induction hypothesis,
	we easily see
	\begin{equation*}
		\tG_{f;n,m}^c (z_1,\dots, z_{n+m}) = G_{f;n,m}^c (z_1,\dots, z_{n+m}).
	\end{equation*}
	Thus the conclusion holds.
\end{proof}

Now by taking $\bm t^+ = \bm t^- =0$ in the above proposition
and the concrete examples presented in the previous subsection,
we obtain the following
\begin{Corollary}
	\label{cor-rel-npt-F}
	The connected bosonic $(n,m)$-point function is given by
	\begin{align*}
		 \big\langle \alpha\big(z_{[n+m]}\big) \big\rangle_{f;n,m}^c ={}&
		i_{z_1,z_2} \frac{\delta_{n,2}\delta_{m,0} }{(z_1-z_2)^2} +
		i_{z_1,z_2} \frac{ \delta_{n,0}\delta_{m,2}}{(z_1-z_2)^2} \\
		&{}+ \sum_{j_1,\dots,j_n,k_1,\dots,k_m \geq 1}
		\frac{\pd^{m+n} F_f(\bm t^+,\bm t^-)}{\pd t_{j_1}^+ \cdots \pd t_{j_n}^+ \pd t_{k_1}^- \cdots \pd t_{k_m}^- }
		\bigg|_{\bm t=0} \cdot
		\prod_{a=1}^n z_a^{-j_a-1} \cdot \prod_{b=1}^m z_{n+b}^{k_b-1}.
	\end{align*}
\end{Corollary}

\subsection[A formula for the connected bosonic (n,m)-point function]{A formula for the connected bosonic $\boldsymbol{(n,m)}$-point function}

Now we are able to present our main result of this section.
First recall the following combinatorial result (see Zhou~\cite[Proposition 5.2]{zhou1})
\begin{Proposition}
	[\cite{zhou1}]
Let $\{\varphi(\xi_1,\dots,\xi_n)\}_{n\geq 1}$
	and $\{\varphi^c(\xi_1,\dots,\xi_n)\}_{n\geq 1}$ be two sequences of functions that related to each other by M\"obius inversion:
	\begin{gather*}
			\varphi(\xi_1,\dots,\xi_n) = \sum_{I_1\sqcup \cdots \sqcup I_k =[n]}
			\frac{1}{k!} \varphi^c (\xi_{I_1})\cdots \varphi^c(\xi_{I_k}),\\
			\varphi^c(\xi_1,\dots,\xi_n) = \sum_{I_1\sqcup \cdots \sqcup I_k =[n]}
			\frac{(-1)^{k-1}}{k} \varphi (\xi_{I_1})\cdots \varphi(\xi_{I_k}),
		\end{gather*}
	where $I_1,\dots,I_k$ are nonempty,
	and $\xi_I = (\xi_i)_{i\in I}$.
	If $\big\{\varphi\big(\xi_{[n]}\big)\big\}_{n\geq 1}$ are of the form
	\begin{equation*}
		\varphi(\xi_1,\dots,\xi_n) = \det(M(\xi_i,\xi_j))_{1\leq i,j\leq n},
	\end{equation*}
	for some function $M(\xi,\eta)$,
	then $\big\{\varphi^c\big(\xi_{[n]}\big)\big\}_{n\geq 1}$ are given by
	\begin{equation*}
		\varphi^c (\xi_1,\dots,\xi_n) = (-1)^{n-1} \sum_{\text{$n$-cycles}}
		\prod_{i=1}^n M\big(\xi_{\sigma(i)}, \xi_{\sigma(i+1)}\big),
	\end{equation*}
	where the summations are taken over $n$-cycles $\sigma$,
	and we denote $\sigma(n+1) = \sigma(1)$.
\end{Proposition}

Combine this proposition with Theorem~\ref{thm-disconn-bos-det},
then we immediately see that the connected bosonic $(n,m)$-point functions associated to $\tau_f$
are given by
\[
\big\langle \alpha\big(z_{[n+m]}\big) \big\rangle_{f;n,m}^c
= (-1)^{n+m-1} \sum_{\text{$(n+m)$-{\rm cycles}}}
\prod_{i=1}^{n+m} B_{\sigma(i), \sigma(i+1)},
\]
where $B_{i,j}$ are given by~\eqref{eq-matrixcoeff-bij}.
Then by Corollary~\ref{cor-rel-npt-F},
we conclude that
\begin{Theorem}
	\label{thm-mainthm-conn}
Suppose $F_f (\bm t^+ , \bm t^- ) = \log \tau_f (\bm t^+ , \bm t^- )$ is the free energy associated to
	the tau-function~$\tau_f$,
	then
	\begin{gather}
		 \sum_{j_1,\dots,j_n,k_1,\dots,k_m \geq 1}
		\frac{\pd^{m+n} F_f(\bm t^+,\bm t^-)}{\pd t_{j_1}^+ \cdots
			\pd t_{j_n}^+ \pd t_{k_1}^- \cdots \pd t_{k_m}^- }
		\bigg|_{\bm t=0} \cdot
		\prod_{a=1}^n z_a^{-j_a-1} \cdot \prod_{b=1}^m z_{n+b}^{k_b-1} \nonumber\\
		\qquad\qquad{}= (-1)^{n+m-1} \sum_{\text{$(n+m)$-{\rm cycles}}}
		\prod_{i=1}^{n+m} B_{\sigma(i), \sigma(i+1)}
		- i_{z_1,z_2} \frac{\delta_{n,2}\delta_{m,0} + \delta_{n,0}\delta_{m,2} }{(z_1-z_2)^2} ,\label{eq-mainthm-sumcycles}
	\end{gather}	
 where the summations are taken over $(n+m)$-cycles $\sigma$,
	and we denote $\sigma(n+m+1) = \sigma(1)$.
	The $(n+m)\times (n+m)$ matrix $(B_{i,j})$ are given by
	\[
	B_{i,j} = \begin{cases}
		A_0 (z_i,z_j) & \text{if} \ i< j\leq n \ \text{or} \ n<i< j,\\
		A_{-f} (z_i,z_j) & \text{if} \ i\leq n<j,\\
		-A_0(z_j,z_i) & \text{if} \ j<i\leq n \ \text{or} \ n<j< i,\\
		-A_{f(-\cdot)}(z_j,z_i) & \text{if} \ j\leq n<i.
	\end{cases}
	\]
\end{Theorem}

A straightforward consequence of the above formula is
\begin{Corollary}
	\label{cor-vanishing}
	One has
	\begin{equation*}
		\sum_{j_1,\dots,j_n,k_1,\dots,k_m \geq 1}
		\frac{\pd^{m+n} F_f(\bm t^+,\bm t^-)}{\pd t_{j_1}^+ \cdots
			\pd t_{j_n}^+ \pd t_{k_1}^- \cdots \pd t_{k_m}^- }
		\bigg|_{\bm t=0} =0
	\end{equation*}
	unless $j_1+ j_2 +\cdots +j_n = k_1+k_2+\cdots +k_m$.
\end{Corollary}
\begin{proof}
	Recall that $A_0(z,w)$, $A_f(z,w)$, and $i_{z,w}\frac{1}{(z-w)^2}$
	are all of the form
	\begin{equation*}
		\sum_{k\geq 0} c_k \cdot z^{-k-1}w^k,
	\end{equation*}
	thus the conclusion is proved by comparing the total orders of non-negative powers and negative powers
	in the right-hand side of~\eqref{eq-mainthm-sumcycles}.
\end{proof}

Furthermore, we have
\begin{Corollary}
	\label{cor-twosides}
	For every $n>0$ or $m>0$, we have
	\begin{gather*}
		 \frac{\pd^{n} F_f(\bm t^+,\bm t^-)}{\pd t_{j_1}^+ \cdots \pd t_{j_n}^+}
		\bigg|_{\bm t=0} =0,
		\qquad \forall j_1,\dots,j_n \geq 1, \\
		 \frac{\pd^{m} F_f(\bm t^+,\bm t^-)}{\pd t_{k_1}^- \cdots \pd t_{k_m}^-}
		\bigg|_{\bm t=0} =0,
		\qquad \forall k_1,\dots,k_m \geq 1.
	\end{gather*}
\end{Corollary}
\begin{proof}
	This is a straightforward consequence of Corollary~\ref{cor-vanishing}.
\end{proof}

In the rest of this subsection,
we give some examples of the formula~\eqref{eq-mainthm-sumcycles}.
We only need to consider the cases where $n>0$ and $m>0$ due to the above corollary.
\begin{Example}\label{exa: (1,1)}
	For $(n,m)=(1,1)$,
	there is only one $2$-cycle $\sigma=(12)$,
	thus
	\begin{equation*}
		\sum_{j,k\geq 1}
		\frac{\pd^2 F_f (\bm t^+,\bm t^-)}{\pd t_j^+ \pd t_k^-}
		\bigg|_{\bm t=0} \cdot z_1^{-j-1} z_2^{k-1}
		= -B_{1,2}B_{2,1}
		=A_{-f}(z_1,z_2)A_{f(-\cdot)} (z_1,z_2).
	\end{equation*}
\end{Example}

\begin{Example}
	For $(n,m)=(2,1)$,
	there are two $3$-cycles $\sigma=(123),(132)$,
	thus
	\begin{gather*}
			\sum_{j,k,l\geq 1}
			\frac{\pd^2 F_f (\bm t^+,\bm t^-)}{\pd t_j^+ \pd t_k^+ \pd t_l^-}
			\bigg|_{\bm t=0} \cdot z_1^{-j-1} z_2^{-k-1} z_3^{l-1}
			= B_{1,2}B_{2,3}B_{3,1} + B_{1,3}B_{3,2}B_{2,1}\\
			\qquad {}=
			-A_0(z_1,z_2) A_{-f}(z_2,z_3) A_{f(-\cdot)}(z_1,z_3)
			+ A_{-f}(z_1,z_3) A_{f(-\cdot)}(z_2,z_3) A_0(z_1,z_2).
	\end{gather*}
	And similarly, for $(n,m)=(1,2)$, we have
	\begin{gather*}
			\sum_{j,k,l\geq 1}
			\frac{\pd^2 F_f (\bm t^+,\bm t^-)}{\pd t_j^+ \pd t_k^- \pd t_l^-}
			\bigg|_{\bm t=0} \cdot z_1^{-j-1} z_2^{k-1} z_3^{l-1}\\
			\qquad {}=
			-A_{-f}(z_1,z_2) A_{0}(z_2,z_3) A_{f(-\cdot)}(z_1,z_3)
			+ A_{-f}(z_1,z_3) A_{0}(z_2,z_3) A_{f(-\cdot)}(z_1,z_2).
	\end{gather*}
\end{Example}

\begin{Example}
	\label{ex-conn-4pt}
	For $(n,m)=(3,1)$,
	there are six $4$-cycles $\sigma=(1234)$, $(1243)$, $(1324)$, $(1342)$, $(1423)$, $(1432)$,
	thus we have
	\begin{gather*}
			\sum_{i,j,k,l\geq 1}
			\frac{\pd^2 F_f (\bm t^+,\bm t^-)}{\pd t_i^+ \pd t_j^+ \pd t_k^+ \pd t_l^-}
			\bigg|_{\bm t=0} \cdot z_1^{-i-1} z_2^{-j-1} z_3^{-k-1} z_4^{l-1}\\
			\qquad{}= -B_{1,2}B_{2,3}B_{3,4}B_{4,1} - B_{1,2}B_{2,4}B_{4,3}B_{3,1} -B_{1,3}B_{3,2}B_{2,4}B_{4,1}\\
			\qquad\quad{} - B_{1,3}B_{3,4}B_{4,2}B_{2,1} -B_{1,4}B_{4,2}B_{2,3}B_{3,1} - B_{1,4}B_{4,3}B_{3,2}B_{2,1} \\
			\qquad{}= A_{0}(z_1,z_2) A_{0}(z_2,z_3) A_{-f}(z_3,z_4) A_{f(-\cdot)}(z_1,z_4)\\
			\qquad\quad{} - A_{0}(z_1,z_2) A_{-f}(z_2,z_4) A_{f(-\cdot)}(z_3,z_4) A_{0}(z_1,z_3)\\
			\qquad\quad{} - A_{0}(z_1,z_3) A_{0}(z_2,z_3) A_{-f}(z_2,z_4) A_{f(-\cdot)}(z_1,z_4)\\
			\qquad\quad{} - A_{0}(z_1,z_3) A_{-f}(z_3,z_4) A_{f(-\cdot)}(z_2,z_4) A_{0}(z_1,z_2)\\
			\qquad\quad{} - A_{-f}(z_1,z_4) A_{f(-\cdot)}(z_2,z_4) A_{0}(z_2,z_3) A_{0}(z_1,z_3)\\
			\qquad\quad{} +A_{-f}(z_1,z_4) A_{f(-\cdot)}(z_3,z_4) A_{0}(z_2,z_3) A_{0}(z_1,z_2).
	\end{gather*}
	Similarly,
	for $(n,m) = (1,3)$, we have
	\begin{gather*}
			\sum_{i,j,k,l\geq 1}
			\frac{\pd^2 F_f (\bm t^+,\bm t^-)}{\pd t_i^+ \pd t_j^- \pd t_k^- \pd t_l^-}
			\bigg|_{\bm t=0} \cdot z_1^{-i-1} z_2^{j-1} z_3^{k-1} z_4^{l-1}\\
			\qquad{}= A_{-f}(z_1,z_2) A_{0}(z_2,z_3) A_{0}(z_3,z_4) A_{f(-\cdot)}(z_1,z_4)\\
			\qquad\quad{} - A_{-f}(z_1,z_2) A_{0}(z_2,z_4) A_{0}(z_3,z_4) A_{f(-\cdot)}(z_1,z_3)\\
			\qquad\quad{} - A_{-f}(z_1,z_3) A_{0}(z_2,z_3) A_{0}(z_2,z_4) A_{f(-\cdot)}(z_1,z_4)\\
			\qquad\quad{} - A_{-f}(z_1,z_3) A_{0}(z_3,z_4) A_{0}(z_2,z_4) A_{f(-\cdot)}(z_1,z_2)\\
			\qquad\quad{} - A_{-f}(z_1,z_4) A_{0}(z_2,z_4) A_{0}(z_2,z_3) A_{f(-\cdot)}(z_1,z_3)\\
			\qquad\quad{} + A_{-f}(z_1,z_4) A_{0}(z_3,z_4) A_{0}(z_2,z_3) A_{f(-\cdot)}(z_1,z_2).
	\end{gather*}
	And for $(n,m) = (2,2)$, we have
	\begin{gather}	
		\sum_{i,j,k,l\geq 1}
		\frac{\pd^2 F_f (\bm t^+,\bm t^-)}{\pd t_i^+ \pd t_j^+ \pd t_k^- \pd t_l^-}
		\bigg|_{\bm t=0} \cdot z_1^{-i-1} z_2^{-j-1} z_3^{k-1} z_4^{l-1} \nonumber\\
		\qquad{} = A_{0}(z_1,z_2) A_{-f}(z_2,z_3) A_{0}(z_3,z_4) A_{f(-\cdot)}(z_1,z_4) \nonumber\\
		\qquad\quad{} - A_{0}(z_1,z_2) A_{-f}(z_2,z_4) A_{0}(z_3,z_4) A_{f(-\cdot)}(z_1,z_3) \nonumber\\
		\qquad\quad{} - A_{-f}(z_1,z_3) A_{f(-\cdot)}(z_2,z_3) A_{-f}(z_2,z_4) A_{f(-\cdot)}(z_1,z_4) \nonumber\\
		\qquad\quad{} - A_{-f}(z_1,z_3) A_{0}(z_3,z_4) A_{f(-\cdot)}(z_2,z_4) A_{0}(z_1,z_2) \nonumber\\
		\qquad\quad{} - A_{-f}(z_1,z_4) A_{f(-\cdot)}(z_2,z_4) A_{-f}(z_2,z_3) A_{f(-\cdot)}(z_1,z_3) \nonumber\\
		\qquad\quad{} + A_{-f}(z_1,z_4) A_{0}(z_3,z_4) A_{f(-\cdot)}(z_2,z_3) A_{0}(z_1,z_2).\label{eq-f-(2,2)}
	\end{gather}
\end{Example}

\begin{Remark}
	It is worth mentioning that here our strategy for deriving the connected $(n,m)$-point functions is different from the method used by Johnson~\cite{jo}.
	In~\cite[Section~3]{jo},
	Johnson computed the double Hurwitz numbers
	(a special example of the diagonal tau-functions of the~2d Toda lattice hierarchy, see Section~\ref{sec-app-Hurwitz} for more details)
	by commuting the (bosonic) operators,
	and his formula is a summation over some commutation patterns.
	And as a consequence,
	his formula depends on the specific choice of the chamber $\mu^\pm$ lie in.
	Now in this paper, we compute directly using the anti-commutation relations of free fermions,
	and our formula~\eqref{eq-intro-main} does not require specifying the partitions $\mu^\pm$
	or listing out the commutation patterns.
\end{Remark}

\begin{Remark}
	Another formula calculating the connected correlators of a diagonal tau-function of the 2d Toda lattice hierarchy can be found in~\cite[Proposition~3.6 and Theorem~5.3]{bdks}.
	Their method, as they stated in their abstract, essentially dealt with the hypergeometric tau-functions of the KP hierarchy
	since they treated one of the two families $\mathbf{t}^+$ and $\mathbf{t}^-$ as time variables
	and the other as parameters.
	Their formula depends on complicated actions of certain operators and summation over graphs.
	It would be interesting to compare their formula with ours and also the formula derived by Zhou in~\cite{zhou1} dealing with tau-functions
	of the KP hierarchy.
\end{Remark}

\section{Reduction to tau-functions of KP hierarchy}
\label{sec-red-KP}

In this section,
we fix $\bm t^-$ and regard a diagonal tau-function $\tau_f(\bm t^+,\bm t^-)$
as a tau-function of the KP hierarchy with time variable $\bm t^+$,
and compute the KP-affine coordinates of this tau-function.

\subsection{Tau-functions of the KP hierarchy and affine coordinates}

In this subsection,
we recall some preliminaries of the affine coordinates of a tau-function
of the~KP hierarchy,
see~\cite{by, zhou1}.

Let $\tau(\bm t)$ be a tau-function of the KP hierarchy satisfying the initial value condition ${\tau (\bm 0) = 1}$,
where $\bm t = (t_1,t_2,t_3,\dots)$ are the KP-time variables.
In Sato's theory,
such a tau-function corresponds to a point in the big cell of the Sato Grassmannian,
and can be specified by the affine coordinates $\{a_{n,m}\}_{n,m\geq 0}$ on the big cell,
see~\cite[Section~3]{zhou1} for details.

The tau-function $\tau(\bm t)$ can be represented in a simple fashion
in terms of its affine coordinates $\{a_{n,m}\}_{n,m\geq 0}$.
In the fermionic Fock space,
the tau-function can be represented as a Bogoliubov transform of the fermionic vacuum:
\begin{equation*}
	\tau = \exp\bigg( \sum_{n,m\geq 0} a_{n,m}\psi_{-m-\half} \psi_{-n-\half}^* \bigg) |0\rangle.
\end{equation*}
This Bogoliubov transform only involves fermionic creators,
and such a representation is unique.
And in the bosonic Fock space,
the tau-function is a linear summation of the Schur functions,
where the coefficients are given by the determinants of affine coordinates,{\samepage
\begin{equation*}
	\tau(\bm t) = \sum_\mu (-1)^{n_1+\cdots +n_k}\cdot
	\det (a_{n_i,m_j})_{1\leq i,j\leq k} \cdot s_\mu(\bm t),
\end{equation*}
where $(m_1,\dots,m_k|n_1,\dots,n_k)$ is the Frobenius notation of the partition $\mu$.}

In~\cite[Section~5]{zhou1},
Zhou has derived a formula for connected $n$-point functions of a tau-function of the KP hierarchy
in terms of its KP-affine coordinates.
Once we know the affine coordinates of a tau-function $\tau(\bm t)$,
one can compute the free energy $\log\tau(\bm t)$ using Zhou's formula.

\begin{Remark}
	In the case of the BKP hierarchy,
	one also has similar results.
	In the fermionic Fock space of type $B$,
	a tau-function $\tau(\bm t)$ of the BKP hierarchy satisfying $\tau(\bm 0) = 1$
	can be represented as a Bogoliubov transform of the vacuum which only involves
	(neutral) fermionic creators.
	In the bosonic Fock space,
	$\tau(\bm t)$ is a linear summation of the Schur $Q$-functions,
	and the coefficients are Pfaffians of BKP-affine coordinates.
	See~\cite{wy} for details
	and a BKP generalization of Zhou's formula for connected $n$-point functions.
\end{Remark}

\subsection{KP-Affine coordinates of a diagonal tau-function}

It is known that a tau-function $\tau_n(\bm t^+ ,\bm t^-)$ of the 2d Toda lattice hierarchy
is a tau-function of the KP hierarchy with KP-time variables $\bm t^+$ (resp.~$\bm t^-$)
for fixed $\bm t^-$ \big(resp.~$\bm t^+$\big) and~$n$.
In this subsection,
we compute the KP-affine coordinates of a diagonal tau-function $\tau_f(\bm t^+,\bm t^-)$.
Here we regard $\bm t^-$ as parameters and $\bm t^+$ as KP-time variables.

Now let
$f\colon \bZ +\half \to \bC$
be an arbitrary function on the set of half-integers,
and let $\hf$ and $\tau_f$ be given by~\eqref{eq-def-hatf} and~\eqref{eq-def-tauf}.
In what follows,
we expand the vector
\begin{equation*}
	{\rm e}^\hf \Gamma_-(\bm t^-)|0\rangle \in \cF^{(0)}
\end{equation*}
as a linear summation of the basis vectors $\{|\mu\rangle \}$.
First by~\eqref{eq-Y--expansion}, we know that
\begin{equation*}
	\Gamma_-(\bm t^-) |0\rangle
	= \sum_\mu s_\mu(\bm t^- ) |\mu\rangle.
\end{equation*}
From the definition~\eqref{eq-def-hatf} we easily see $\hf |0\rangle = 0$, and then $ {\rm e}^\hf |0 \rangle = |0\rangle$.
Now by~\eqref{eq-vectormu-psi} and the Baker--Campbell--Hausdorff formula~\eqref{eq-conj-f-psi},
we have
\begin{equation*}
	\begin{split}
		{\rm e}^\hf |\mu\rangle
		={}& (-1)^{\sum_{j=1}^k n_j} \cdot {\rm e}^\hf
		\psi_{-m_1-\half} \psi_{-n_1-\half}^* \cdots
		\psi_{-m_k-\half} \psi_{-n_k-\half}^* |0\rangle \\
		={}& (-1)^{\sum_{j=1}^k n_j} \cdot
		\big({\rm e}^\hf \psi_{-m_1-\half} {\rm e}^{-\hf} \big)
		\big({\rm e}^\hf \psi_{-n_1-\half}^* {\rm e}^{-\hf} \big) \cdots
		\big({\rm e}^\hf \psi_{-m_k-\half} {\rm e}^{-\hf} \big)
		\big({\rm e}^\hf \psi_{-n_k-\half}^* {\rm e}^{-\hf} \big) |0\rangle \\
		={}&
		{\rm e}^{\sum_{j=1}^k f(-m_j-\half) -\sum_{j=1}^k f(n_j+\half) } |\mu\rangle,
	\end{split}
\end{equation*}
where $\mu = (m_1,\dots,m_k| n_1,\dots,n_k)$ is the Frobenius notation of the partition $\mu$.
And thus
\begin{align*}
		{\rm e}^\hf \Gamma_- (\bm t^-) |0\rangle
		&= \sum_\mu s_\mu (\bm t^-) {\rm e}^\hf |\mu\rangle
		\\
		&= \sum_\mu
		{\rm e}^{\sum_{j=1}^k f(-m_j-\half) -\sum_{j=1}^k f(n_j+\half) }
		\cdot s_\mu (\bm t^-)
		|\mu\rangle.
	\end{align*}

Recall that if we expand a vector $|U\rangle \in \cF$ as an (infinite) linear combination
of the basis vectors $|\mu\rangle$ of the fermionic Fock space,
then the KP-affine coordinate $a_{n,m}$ of the tau-function
$\langle 0| \Gamma_+(\bm t) |U \rangle$
is exactly $(-1)^n$ times the coefficient of the vectors $|(m|n)\rangle$,
for every $m,n\geq 0$.
Therefore we conclude that
\begin{Theorem}
	\label{thm-KPaffine-diag}
	The KP-affine coordinates for
	$\tau_f (\bm t^+,\bm t^-)$ \big(with fixed $\bm t^-$ and KP-time variables~$\bm t^+$\big) are
	\begin{equation}	
 \label{eq-KPaffine-diag+}
	a_{n,m}^f = (-1)^n \cdot
	s_{(m|n)}(\bm t^-)
	\cdot {\rm e}^{ f(-m-\half) - f(n+\half) },
	\end{equation}	
 for every $m,n\geq 0$.
\end{Theorem}

Now one is able to plugging the above affine coordinates~\eqref{eq-KPaffine-diag+}
into Zhou's formula~\cite[Section~5]{zhou1}
to compute the connected $n$-point functions of $\tau_{f} (\bm t^+,\bm t^-)$.

\subsection[Restriction to t\_k\^{}- = delta\_\{k,1\}]{Restriction to $\boldsymbol{t_k^- = \delta_{k,1}}$}

In this subsection,
we consider the special evaluation $\bm t^- = (1,0,0,0,\dots)$
of the diagonal tau-function $\tau_f (\bm t^+,\bm t^-)$.
This will be useful in the computation of KP-affine coordinates of
the tau-function of single Hurwitz numbers,
see Section~\ref{sec-singleH}.

Let
$f\colon \bZ +\half \to \bC$
be a function on the set of half-integers,
and let $\hf$ be the operator given by~\eqref{eq-def-hatf}.
Define
\[
\tilde\tau_f (\bm t) = \langle 0| \Gamma_+(\bm t) \exp\big(\hf\big) \exp( \alpha_{-1}) |0\rangle,
\]
then $\tilde\tau_f (\bm t)$ is a tau-function of the KP hierarchy.
By Theorem~\ref{thm-KPaffine-diag}, we know that
the KP-affine coordinates of $\tilde\tau_f (\bm t)$ is
\begin{equation*}
	(-1)^n \cdot
	s_{(m|n)}(\delta_{k,1})
	\cdot {\rm e}^{ f(-m-\half) - f(n+\half) },
\end{equation*}
where $s_\mu(\delta_{k,1})$ means evaluating the Schur function $s_\mu(\bm t)$ at time $\bm t = (1,0,0,0,\dots)$.

Now we consider the evaluation $s_\mu(\delta_{k,1})$.
The following identity is well known in literatures
(see, e.g.,~\cite[Section~4.1]{fh}):
\[
s_\mu (\delta_{k,1})
= \frac{1}{l_1! \cdots l_k!} \cdot \prod_{i<j} (l_i-l_j),
\]
where for a partition $\mu =  (\mu_1 \geq \mu_2 \geq\cdots\geq \mu_{l(\mu)} >0 )$ we denote
\[
l_i = \mu_i +l(\mu)-i, \qquad
i=1,\dots,l(\mu).
\]
Or more explicitly,
\[
s_\mu (\delta_{k,1})
= \sum_\mu
\frac{\prod_{1\leq i<j\leq l(\mu)} (\mu_i-\mu_j -i+j)}{\prod_{i=1}^{l(\mu)}(\mu_i +l(\mu)-i)!}.
\]
Now take $\mu$ to be the hook partition $\mu = (m|n)$,
i.e., $l(\mu) = n+1 $,
and $\mu_1 = m+1$, $\mu_2= \cdots =\mu_{n+1} =1$.
Then, we have
\begin{align*}
		s_{(m|n)} (\delta_{k,1})& =
		\frac{\prod_{1\leq i<j\leq n+1} (\mu_i-\mu_j -i+j)}{\prod_{i=1}^{n+1}(\mu_i +n+1-i)!}
		= \frac{(m+n)! \cdot \prod_{j=1}^{n-1} j^{n-j}}{m!\cdot (m+n+1) \cdot \prod_{j=1}^{n} j!}
		\\
		&= \frac{(m+n)!}{(m+n+1)\cdot m!\cdot n!}.
\end{align*}
Therefore, we conclude that
\begin{Proposition}
	\label{prop-restri-t-=1}
	The KP-affine coordinates for the tau-function
	$\tilde\tau_f (\bm t)$ are
	\[	
 \tilde a_{n,m}^f = (-1)^n \cdot
	\frac{(m+n)!}{(m+n+1)\cdot m!\cdot n!}
	\cdot {\rm e}^{f(-m-\half) - f(n+\half) },
	\]	
 for every $m,n\geq 0$.
\end{Proposition}

\section{Application to connected double Hurwitz numbers}\label{sec-app-Hurwitz}

In this section,
we use the above results to derive explicit formulas for various types of connected double Hurwitz numbers,
including the ordinary double Hurwitz numbers,
the double Hurwitz numbers with completed $r$-cycles,
and mixed double Hurwitz numbers.

\subsection{Ordinary double Hurwitz numbers and the associated tau-function}

First, we recall some facts of double Hurwitz numbers in literatures.

Let $\mu^+$ and $\mu^-$ be two partitions of a positive integer $d$.
Consider a branched cover
$f\colon \Sigma \to \bP^1$
from a smooth Riemann surface $\Sigma$ of genus $g$ to the complex projective line,
with ramification types $\mu^+$ over $0\in \bP^1$ and $\mu^-$ over $\infty \in \bP^1$,
and $\big(1^{d-2} 2\big)$ (i.e., simple ramification) over other $b$ fixed points.
Then by the Riemann--Hurwitz formula, one has
\begin{equation}
\label{eq-RiemHurw}
2g-2+ l(\mu^+) +l(\mu^-) =b,
\end{equation}
where $l(\mu)$ denotes the length of a partition $\mu$.
Two such covers $f\colon \Sigma \to \bP^1$ and $f'\colon \Sigma'\to\bP^1$ are said to be equivalent,
if there exists a biholomorphic map $\phi\colon \Sigma\to \Sigma'$
such that $f = f'\circ \phi$.
The map $\phi$ is called an automorphism of $f$.

The possibly connected double Hurwitz numbers $H_g^\bullet (\mu^+,\mu^-)$
is defined to be the following weighted counting of the equivalence classes of such maps:
\begin{equation*}
	H_g^\bullet (\mu^+,\mu^-) = \sum_f \frac{1}{|\Aut (f)|},
\end{equation*}
where the Riemann surface $\Sigma$ is possibly disconnected.
And the connected double Hurwitz numbers $H_g^\circ (\mu^+,\mu^-)$ is the weighted counting of
the equivalence classes of maps $f$ from a~connected Riemann surface $\Sigma$:
\begin{equation*}
	H_g^\circ (\mu^+,\mu^-) = \sum_{\text{$f$: connected}} \frac{1}{|\Aut (f)|}.
\end{equation*}
Notice that when $g$ and $\mu^\pm$ are fixed,
the number $b$ is determined by~\eqref{eq-RiemHurw}.

In~\cite{Ok1},
Okounkov has shown that the generating series
\[
\tau^{(2)} \big( \bm t^+ , \bm t^-; \beta\big) = \sum_{g,\mu^\pm}
\frac{ \beta^b p_{\mu^+}^+ p_{\mu^-}^- H_g^\bullet (\mu^+,\mu^-)}{b!}
\]
of all possibly disconnected double Hurwitz numbers is a tau-function of the 2d Toda lattice hierarchy,
where $\bm p^\pm = \big(p_1^\pm,p_2^\pm, p_3^\pm,\dots\big)$ are two sequences of formal variables,
and we denote
\begin{equation*}
	p_\mu^\pm = p^\pm_{\mu_1} p^\pm_{\mu_2} \cdots p^\pm_{\mu_l}
\end{equation*}
for a partition $\mu = (\mu_1,\mu_2,\dots,\mu_l)$.
The time variables of this hierarchy are
\begin{equation*}
	t_n^\pm = \frac{1}{n}p_n^\pm,
	\qquad n\geq 1.
\end{equation*}
He derived the following fermionic representation of this tau-function:
\[
\tau^{(2)} \big( \bm t^+ , \bm t^-; \beta\big) =
\big\langle 0 \big| \Gamma_+(\bm t^+) {\rm e}^{\beta K^{(2)}} \Gamma_- (\bm t^-) \big|0\big\rangle,
\]
where $K^{(2)}$ is the cut-and-join operator~\eqref{eq-def-C&Jopr}.
If we regard $p_n^\pm = p_n \big(\bm x^\pm\big)$ to be the Newton symmetric function of degree $n$
in some formal variables $\bm x^\pm = \big(x_1^\pm,x_2^\pm,\dots\big)$,
then the following expansion by Schur functions
follows from~\eqref{eq-Y--expansion} and~\eqref{eq-eigen-C&J}:
\begin{equation*}
	\tau^{(2)} \big( \bm t^+ , \bm t^-; \beta\big)
	= \sum_\mu {\rm e}^{\beta \kappa_\mu /2} s_\mu \big(\bm x^+\big) s_\mu (\bm x^- ).
\end{equation*}

\subsection{Computation of ordinary connected double Hurwitz numbers}

The free energy $F^{(2)} = \log \tau^{(2)}$ associated to the tau-function $\tau^{(2)}$
is the generating series of connected double Hurwitz numbers:
\begin{align*}
	F^{(2)} ( \bm t^+ , \bm t^-; \beta ) = \sum_{g,\mu^\pm}
	\frac{ \beta^b p_{\mu^+}^+ p_{\mu^-}^- H_g^\circ (\mu^+,\mu^-)}{b!}
	= \sum_{\mu^+,\mu^-} H^\circ (\mu^+,\mu^-; \beta )
	p_{\mu^+}^+ p_{\mu^-}^-,
\end{align*}
where we denote
\[
H^\circ (\mu^+,\mu^-; \beta )
= \sum_g \frac{\beta^b}{b!} H_g^\circ (\mu^+,\mu^- ),
\]
and $b$ is determined by~\eqref{eq-RiemHurw}.
Then, we are able to apply the results derived in Section~\ref{sec-conn-nmpt}
to compute $H^\circ (\mu^+,\mu^-;\beta )$.
In this case,
the function~\eqref{eq-def-functionf} is taken to be
\[
f^{(2)}(s) = \beta \cdot \frac{s^2}{2},
\qquad \forall s\in \bZ+\half,
\]
then the operator~\eqref{eq-def-hatf} becomes $\hat f^{(2)} = \beta K^{(2)}$, and
\begin{gather}
	 A_{-f^{(2)}} (z,w) = \sum_{k=0}^\infty
	{\rm e}^{-\half\beta (k+\half)^2}\cdot z^{-k-1} w^k,\nonumber\\
	 A_{f^{(2)}(-\cdot)} (z,w) =
	\sum_{k=0}^\infty {\rm e}^{\half\beta (k+\half)^2}\cdot z^{-k-1} w^k.\label{eq-Af-(2)}
\end{gather}
Then by Theorem~\ref{thm-mainthm-conn}, we have
\begin{Theorem}
	Let $\mu^+ = (\mu_1^+,\dots,\mu_n^+ )$ and $\mu^- = (\mu_1^-,\dots,\mu_m^-)$ be
	two partitions with $|\mu^+| = |\mu^-|$,
then we have
	\begin{gather}	
 H^\circ (\mu^+,\mu^-; \beta )\nonumber\\
 \qquad{}
		=
		\frac{1}{Z_{\mu^+}Z_{\mu^-}}
		\operatorname{Coeff}_{\prod_{a=1}^n z_a^{-\mu_a^+ -1} \prod_{b=1}^m z_{n+b}^{\mu_b^- -1}}
		\!\Bigg[
		(-1)^{n+m-1}\!
 \sum_{\text{$(n+m)$-{\rm cycles}}}
		\prod_{i=1}^m B^{(2)}_{\sigma(i),\sigma(i+1)}
		\Bigg], \!\!\! \label{eq-thm-hur-nmpt}
	\end{gather}	
 where $\operatorname{Coeff}$ means taking the coefficient,
	and
	\begin{equation*}
		Z_\mu = \prod_{j\geq 0} m_j(\mu)! \cdot j^{m_j(\mu)}
	\end{equation*}
	for a partition $\mu = \bigl(1^{m_1(\mu)} 2^{m_2(\mu)} \cdots\bigr)$.
	And $B^{(2)}_{i,j}$ are given by
	\begin{gather*}
		B_{i,j}^{(2)} = \begin{cases}
			i_{z_i,z_j} \frac{1}{z_i-z_j} & \text{if} \ i< j\leq n \ \text{or} \ n<i< j,\\
			A_{-f^{(2)}} (z_i,z_j) & \text{if} \ i\leq n<j, \\
			i_{z_j,z_i} \frac{1}{z_i-z_j} & \text{if} \ j<i\leq n \ \text{or} \ n<j< i,\\
			-A_{f^{(2)}(-\cdot)}(z_j,z_i) & \text{if} \ j\leq n<i.
		\end{cases}
	\end{gather*}
\end{Theorem}

\begin{Remark}
	Here we do not need to consider the extra terms $-i_{z_1,z_2}\frac{1}{(z_1-z_2)^2}$
	in Theorem~\ref{thm-mainthm-conn} for $(n,m) = (2,0)$ or $(0,2)$
	due to Corollary~\ref{cor-twosides}.
\end{Remark}

\begin{Remark}
	The evaluation of $A_{-f^{(2)}}(z,w)$ and $A_{f^{(2)(-\cdot)}}(z,w)$ at $z=w=1$ are some theta constants:
	\begin{equation*}
		A_{-f^{(2)}}(1,1) = \frac{1}{2}\vartheta(q), \qquad
		A_{f^{(2)(-\cdot)}}(1,1) = \frac{1}{2}\vartheta\big(q^{-1}\big),
	\end{equation*}
	where $q={\rm e}^{-\frac{1}{2}\beta}$,
	and $\vartheta(q)=\sum_{n\in \bZ} q^{(n+1/2)^2}$
	is the evaluation of the theta function $\vartheta_{1,0}(z,q)$ at~${z=0}$.
	As a result,
	the formula for the generating series of the connected $(1,1)$-point functions in Example~\ref{exa: (1,1)} gives
	\begin{equation*}
		\sum_{j,k\geq1}jk \cdot H^\circ\big((j),(k);\beta\big)
		=\frac{1}{4} \vartheta(q) \vartheta\big(q^{-1}\big).
	\end{equation*}
\end{Remark}

In the rest of this subsection,
we give some examples of the formula~\eqref{eq-thm-hur-nmpt} for the ordinary double Hurwitz numbers.
\begin{Example}
	The first non-trivial case is $(n,m) = (1,1)$.
	Then $\mu^+ =\big(\mu^+_1\big)$ and $\mu^-=(\mu^-_1)$ are both of length one,
	and by~\eqref{eq-Af-(2)} we have
	\begin{align*}
			 H^\circ \big(\big(\mu^+_1\big),(\mu^-_1);\beta \big)
		 &{}= \frac{1}{\mu^+_1\mu^-_1}
			\operatorname{Coeff}_{z_1^{-\mu^+_1-1} z_2^{\mu_1^- -1}}
			\big( A_{-f^{(2)}}(z_1,z_2) A_{f^{(2)}(-\cdot)} (z_1,z_2) \big)\\
			&{}= \delta_{\mu_1^+,\mu_1^-} \cdot \frac{1}{\big(\mu_1^+\big)^2}\sum_{k=0}^{\mu_1^+-1}
			{\rm e}^{-\half \beta \mu_1^+ (\mu_1^+ -1-2k)}.
	\end{align*}
	More concretely,
	we have
	\begin{gather*}
			\sum_{b} H^\circ_g ((1),(1))\frac{\beta^b}{b!} = 1,\\
			\sum_{b} H^\circ_g ((2),(2))\frac{\beta^b}{b!} = \frac{{\rm e}^{-\beta} + {\rm e}^{\beta}}{4}
			= \frac{1}{2} + \frac{\beta^2}{4} + \frac{\beta^4}{48}
			+ \frac{\beta^6}{1440} + \frac{\beta^8}{80640}+ \cdots,\\
			\sum_{b} H^\circ_g ((3),(3))\frac{\beta^b}{b!} = \frac{{\rm e}^{-3\beta}+1 + {\rm e}^{3\beta}}{9}
			= \frac{1}{3} + \beta^2 + \frac{3}{4} \beta^4 + \frac{9}{40}\beta^6
			+\cdots.
	\end{gather*}
	In general,
	for a positive integer $u$, we have
	\[
 \sum_{b} H^\circ_g ((u),(u))\frac{\beta^b}{b!}
	= \frac{1}{u^2} {\rm e}^{-\half \beta u (u-1)} \sum_{k=0}^{u-1} ({\rm e}^{\beta u})^k
	= \frac{{\rm e}^{\half \beta u^2} - {\rm e}^{-\half \beta u^2}}{u^2 \big({\rm e}^{ \half \beta u}-{\rm e}^{-\half \beta u}\big)}.
	\]
\end{Example}

\begin{Example}
	For $(n,m)=(2,1)$, we have
	\begin{gather*}
			 H^\circ \big(\big(\mu_1^+,\mu_2^+\big),(\mu_1^-);\beta\big)\\
			\qquad{}= \frac{1}{Z_{\mu^+}Z_{\mu^-}}
			\operatorname{Coeff}_{z_1^{-\mu_1^+-1} z_2^{-\mu_2^+-1} z_3^{\mu_1^--1}}\biggl(
			- \frac{i_{z_1,z_2}}{z_1-z_2} \cdot A_{-f^{(2)}}(z_2,z_3) A_{f^{(2)}(-\cdot)}(z_1,z_3)\\
			 \qquad\quad{} + A_{-f^{(2)}}(z_1,z_3) A_{f^{(2)}(-\cdot)}(z_2,z_3) \cdot \frac{i_{z_1,z_2}}{z_1-z_2}\biggr)\\
			\qquad{}= \frac{\delta_{\mu_1^+ + \mu_2^+ , \mu_1^-}}{\big(1+\delta_{\mu_1^+,\mu_2^+}\big) \mu_1^+ \mu_2^+ \mu_1^- }
			\left( -\sum_{k= \mu_2^+}^{\mu_1^--1}{\rm e}^{\half \beta \mu_1^- (\mu_1^- -1-2k)}
			+ \sum_{k= \mu_2^+}^{\mu_1^- -1} {\rm e}^{-\half \beta\mu_1^- (\mu_1^- -1-2k)} \right).
	\end{gather*}
	For example,
	\begin{gather*}
			 \sum_{b} H^\circ_g ((1,1),(2))\frac{\beta^b}{b!} =
			\frac{{\rm e}^\beta - {\rm e}^{-\beta}}{4} =
			\frac{\beta}{2} + \frac{\beta^3}{12} + \frac{\beta^5}{240} + \frac{\beta^7}{10080}
			+ \frac{\beta^9}{725760}+\cdots,\\
			 \sum_{b} H^\circ_g ((2,1),(3))\frac{\beta^b}{b!} =
			\frac{{\rm e}^{3\beta} - {\rm e}^{-3\beta}}{6} =
			\beta + \frac{3}{2} \beta^3 + \frac{27 }{40}\beta^5 + \frac{81 }{560}\beta^7
			+\cdots,\\
			 \sum_{b} H^\circ_g ((3,1),(4))\frac{\beta^b}{b!} =
			\frac{{\rm e}^{6\beta} - {\rm e}^{-6\beta}}{12} =
			\beta + 6 \beta^3 + \frac{54 }{5}\beta^5 + \frac{324 }{35}\beta^7 +\cdots, \\
			 \sum_{b} H^\circ_g ((2,2),(4))\frac{\beta^b}{b!} =
			\frac{{\rm e}^{6\beta} +{\rm e}^{2\beta} -{\rm e}^{-2\beta} - {\rm e}^{-6\beta}}{32} =
			\frac{\beta}{2} + \frac{7 }{3}\beta^3 + \frac{61 }{15}\beta^5 +\cdots.
	\end{gather*}
	In general, for integers $u\geq v\geq 1$, we have
	\begin{align*}
		\sum_{b} H^\circ_g ((u,v),(u+v))\frac{\beta^b}{b!}
		&{}= \frac{-\sum_{k= v}^{u+v-1}{\rm e}^{\half \beta (u+v) (u+v -1-2k)}
			+ \sum_{k= v}^{u+v -1} {\rm e}^{-\half \beta (u+v) (u+v -1-2k)}}
		{(1+\delta_{u,v}) uv(u+v) \cdot
			\big({\rm e}^{\half \beta (u+v)} - {\rm e}^{-\half \beta (u+v)}\big) } \\
		&{}= \frac{{\rm e}^{\half \beta (u+v)^2} + {\rm e}^{-\half \beta (u+v)^2}
			- {\rm e}^{\half \beta (v^2-u^2)} -{\rm e}^{\half \beta (u^2-v^2)}}
		{(1+\delta_{u,v}) uv(u+v) \cdot
			\big({\rm e}^{\half \beta (u+v)} - {\rm e}^{-\half \beta (u+v)}\big) }.
	\end{align*}	
 In particular, when $u=v$, we have
	\[	
 \sum_{b} H^\circ_g ((u,u),(2u))\frac{\beta^b}{b!}
	= \frac{{\rm e}^{2\beta u^2} + {\rm e}^{-2\beta u^2} -2}
	{4u^3 \big({\rm e}^{\beta u} - {\rm e}^{-\beta u}\big)}.
	\]
\end{Example}

\begin{Example}
	Now we consider the case $(n,m)= (2,2)$.
	Denote $\mu^+ = (u_1,u_2)$ and $\mu^- = (v_1,v_2)$.
	By plugging the expressions~\eqref{eq-Af-(2)} of $A_{-f^{(2)}}$ and $A_{f^{(2)}(-\cdot)}$
	into~\eqref{eq-f-(2,2)}
	(see also Example~\ref{ex-conn-4pt}),
	we obtain the following
	(here we omit the details of computations):
	\begin{gather*}
 \sum_{b} H_g((u_1,u_2),(v_1,v_2)) \frac{\beta^b}{b!}
		\\
\qquad{}= \frac{1}{Z_{\mu^+} Z_{\mu^-}} \frac{\sinh\big(\frac{\beta d^2}{2}\big) + \sinh\big(\frac{\beta d(v_2-v_1)}{2}\big) +
			\sinh\big(\frac{\beta (2v_1u_2-d^2)}{2}\big) + \sinh \big(\frac{\beta(2u_1v_1 -d^2)}{2}\big)}
		{\sinh \big(\half \beta d\big)},
	\end{gather*}
 where
	\begin{equation*}
		\frac{1}{Z_{\mu^+} Z_{\mu^-}} =
		\frac{1}{(1+\delta_{u_1,u_2})(1+\delta_{v_1,v_2})u_1u_2v_1v_1}.
	\end{equation*}
	In particular,
	when $u_1 = v_1$ and $u_2=v_2$
	(i.e., on the walls in the sense of~\cite{gjv, jo})
	we have
	\begin{gather*}
		\sum_{b} H_g((u_1,u_2),(u_1,u_2)) \frac{\beta^b}{b!}		\\
\!\quad{}= \frac{1}{(1+\delta_{u_1,u_2})^2 u_1^2 u_2^2} \frac{\sinh\big(\frac{\beta d^2}{2}\big) + \sinh\big(\frac{\beta d(u_2-u_1)}{2}\big) +
			\sinh\bigr(-\frac{\beta (u_1^2+u_2^2)}{2}\bigl) + \sinh \big(\frac{\beta(2u_1^2 -d^2)}{2}\big)}
		{\sinh \big(\half \beta d\big)}.
	\end{gather*}
\end{Example}

\subsection[Connected double Hurwitz numbers with completed r-cycles]{Connected double Hurwitz numbers with completed $\boldsymbol{r}$-cycles}

In this subsection,
we consider connected double Hurwitz numbers where the simple ramification type $\big(2,1^{d-2}\big)$
is replaced by the completed $r$-cycle.
For an introduction to double Hurwitz numbers with completed $r$-cycles
and the relation to the Gromov--Witten theory of $\bC\bP^1$,
see Okounkov-Pandharipande~\cite{op}.
See also Shadrin--Spitz--Zvonkine~\cite{ssz}.

The generating series of possibly disconnected double Hurwitz numbers
with completed $r$-cycles is the following tau-function
\big(see~\cite[equations~(31) and (33)]{ssz},
and notice here our notation $H_g^{(r)\bullet}(\mu^+,\mu^-)$ differs
from the notation $h_{g,\mu^+,\mu^-}^{(r)}$ in~\cite{ssz}
by an additional factor $l (\mu^+ )!\cdot l(\mu^-)!$\big):
\begin{gather*}
	\tau^{(r)}(\bm t^+,\bm t^-; \beta)
	=
	\sum_{g,\mu^+,\mu^-} \frac{\beta^b p^+_{\mu^+} p^-_{\mu^-} H_g^{(r)\bullet}(\mu^+,\mu^-)}{b!}
	=
	\big\langle 0 \big| \Gamma_+(\bm t^+) {\rm e}^{\beta K^{(r)}} \Gamma_- (\bm t^-) \big|0\big\rangle,
\end{gather*}
where the number $b$ is determined by
\begin{equation*}
	b = \big(2g-2+ l(\mu^+) +l(\mu^-) \big)/r,
\end{equation*}
and $K^{(r)}$ is the following operator on the fermionic Fock space:
\begin{equation}
\label{eq-def-C&J-r}
K^{(r)} = \sum_{s\in \bZ+\half} \frac{s^r}{r!} {:}\psi_s \psi_{-s}^*{:}.
\end{equation}
In this case,
the function~\eqref{eq-def-functionf} is taken to be
\[
f^{(r)}(s) = \beta \cdot \frac{s^r}{r!},
\qquad \forall s\in \bZ+\half,
\]
then the operator~\eqref{eq-def-hatf} in this case is $\hat f^{(r)} = \beta K^{(r)}$.
Similar to~\eqref{eq-thm-hur-nmpt},
we have
\begin{Theorem}
	For two partitions $\mu^+ = (\mu_1^+,\dots,\mu_n^+ )$, $\mu^- = (\mu_1^-,\dots,\mu_m^-)$
	with $|\mu^+| = |\mu^-|$,
	the connected double Hurwitz numbers with completed $r$-cycles are
	\begin{gather*}
		 \sum_b \frac{\beta^b}{b!}
		H_g^{(r)\circ} (\mu^+,\mu^-)\\
\qquad{}		=
		\frac{1}{Z_{\mu^+}Z_{\mu^-}}
		\operatorname{Coeff}_{\prod_{a=1}^n z_a^{-\mu_a^+ -1} \prod_{b=1}^m z_{n+b}^{\mu_b^- -1}}
		\Bigg[
		(-1)^{n+m-1} \sum_{\text{$(n+m)$-{\rm cycles}}}
		\prod_{i=1}^m B^{(r)}_{\sigma(i),\sigma(i+1)}
		\Bigg],
	\end{gather*}
 where
	\begin{equation*}
		B_{i,j}^{(r)} = \begin{cases}
		\displaystyle	i_{z_i,z_j} \frac{1}{z_i-z_j} & \text{if} \ i< j\leq n \ \text{or} \ n<i< j,\\
		\displaystyle	\sum\limits_{k=0}^\infty
			{\rm e}^{- \beta (k+\half)^r /r!}\cdot z_i^{-k-1} z_j^k
			& \text{if} \ i\leq n<j,\\
		\displaystyle	i_{z_j,z_i} \frac{1}{z_i-z_j} & \text{if} \ j<i\leq n \ \text{or} \ n<j< i,\\
		\displaystyle	-\sum\limits_{k=0}^\infty {\rm e}^{\beta (-k-\half)^r /r!}
			\cdot z_j^{-k-1} z_i^k & \text{if} \ j\leq n<i.
		\end{cases}
	\end{equation*}
\end{Theorem}

\subsection{Connected mixed double Hurwitz numbers}

In this subsection,
we apply our formula to the mixed double Hurwitz numbers.

The mixed double Hurwitz numbers are introduced by Goulden, Guay-Paquet, and Novak~\cite{ggn}.
These numbers interpolate combinatorially between the ordinary double Hurwitz numbers
and the monotone double Hurwitz numbers~\cite{ggn2},
and are related to combinatorial aspects of Cayley graph of the symmetric groups $S_n$.
Moreover,
those authors showed that the generating series of these
(disconnected)
mixed double Hurwitz numbers
is a diagonal tau-function solution to the 2d Toda hierarchy.
See~\cite{ggn} for details of the constructions and notations.

Let $\mu^\pm$ be two partitions of an integer $d$,
and let $k,l\geq 0$ be two integers.
Denote by $W^{\bullet k,l}(\mu^+,\mu^-)$ the possibly disconnected mixed double Hurwitz number
indexed by $k$, $l$ and $\mu^\pm$,
and let
\begin{equation*}
	W^\bullet (t,u;\bm t^+,\bm t^- )
	= 1+\sum_{d=1}^\infty \frac{1}{d!}
	\sum_{k,l=0}^\infty \frac{t^k u^l}{l!}
	\sum_{|\mu^\pm|=d}
	W^{\bullet k,l}(\mu^+,\mu^-) p^+_{\mu^+}p^-_{\mu^-},
\end{equation*}
where $p_n^\pm = n\cdot t_n^\pm$.
If one regards $p_n$ as the Newton symmetric function of degree $n$,
then one has the following Schur function expansion
(see~\cite[Section~2]{ggn}):
\[
W^\bullet(t,u;\bm t^+,\bm t^-) =
\sum_{\lambda} Y(\lambda) s_\lambda^+ s_{\lambda}^-,
\]
where the summation is taken over all partitions
(or equivalently, all Young diagrams) $\lambda$,
and
\[
Y(\lambda) = \prod_{ \Box \in \lambda}
\frac{ {\rm e}^{c(\Box) u}}{1-c(\Box) t}.
\]
Here $\Box \in \lambda$ is a box in the Young diagram,
and $c(\Box)$ is the content of this box,
i.e.,
if $\Box$ is in the $i$-th row and $j$-th column,
then $c(\Box) = j-i$.

\begin{Remark}
	When $k=0$, the mixed double Hurwitz numbers are reduced to the ordinary double Hurwitz numbers.
	And when $l=0$, the mixed double Hurwitz numbers are reduced to the monotone double Hurwitz numbers introduced in~\cite{ggn2}.
	See~\cite{ggn} for details.
\end{Remark}

\begin{Lemma}
	Let $f^{\rm mix}\colon \bZ+\half \to \bC$ be the following function:
	\begin{equation*}
		f^{\rm mix} (s) =\begin{cases}
		\displaystyle	\frac{s^2}{2}u - \log \prod\limits_{j=1}^{-s-\half} (1-jt)
			& \text{if} \ s<0,\\
		\displaystyle	\frac{s^2}{2}u + \log \prod\limits_{j=1}^{s-\half} (1+jt)
			& \text{if} \ s>0,
		\end{cases}
	\end{equation*}
	then we have
	\begin{equation*}
		W^\bullet (t,u;\bm t^+,\bm t^-) =
		\big\langle 0\big| \Gamma_+(\bm t^+) \exp\big(\hf^{\rm mix}\big) \Gamma_-(\bm t^-) \big|0\big\rangle.
	\end{equation*}
\end{Lemma}
\begin{proof}
	Since $\Gamma_-(\bm t^-)|0\rangle = \sum_\mu s_\mu |\mu\rangle$ (see~\eqref{eq-Y--expansion}),
	we only need to prove
	\begin{equation}	
 \label{eq-mix-toprove}
	\exp\big(\hf^{\rm mix}\big) |\mu\rangle = Y(\mu) |\mu\rangle
	\end{equation}	
 for every partition $\mu$,
	where $|\mu\rangle$ is the vector~\eqref{eq-cFbasis-mu}.
	
	Assume $\mu=(m_1,\dots, m_k | n_1,\dots,n_k)$,
	then by ${\rm e}^{-\hf^{\rm mix}} |0\rangle = |0\rangle$ and~\eqref{eq-conj-f-psi},
	\begin{align*}
			\exp\big(\hf^{\rm mix}\big) |\mu\rangle
			={}& (-1)^{\sum_{i=1}^k n_i} \cdot
			\big({\rm e}^{\hf^{\rm mix}} \psi_{-m_1-\half} {\rm e}^{-\hf^{\rm mix}}\big)
			\big({\rm e}^{\hf^{\rm mix}} \psi_{-n_1-\half}^* {\rm e}^{-\hf^{\rm mix}}\big) \\
			&{} \cdots ({\rm e}^{\hf^{\rm mix}} \psi_{-m_k-\half} {\rm e}^{-\hf^{\rm mix}})
			\big({\rm e}^{\hf^{\rm mix}} \psi_{-n_k-\half}^* {\rm e}^{-\hf^{\rm mix}}\big) |0\rangle\\
			={}&
			{\rm e}^{\sum_{i=1}^k f^{\rm mix}(-m_i-\half)
				-\sum_{j=1}^k f^{\rm mix}(n_j+\half)} |\mu\rangle.
		\end{align*}
	Notice that
	\begin{equation*}
			Y(\mu) = \prod_{l=1}^k
			\Bigg(\prod_{i=1}^{m_l}\frac{{\rm e}^{iu}}{1-it}\Bigg)
			\Bigg(\prod_{j=1}^{n_l}\frac{{\rm e}^{-ju}}{1+jt}\Bigg),
	\end{equation*}
	and now one easily checks that~\eqref{eq-mix-toprove} holds.
\end{proof}

Then by Theorem~\ref{thm-mainthm-conn}, we have
\begin{Theorem}
	For two partitions $\mu^+ = (\mu_1^+,\dots,\mu_n^+)$, $\mu^- = (\mu_1^-,\dots,\mu_m^-)$
	with ${|\mu^+|\! = \! |\mu^-| \!=\!d}$,
	the connected mixed double Hurwitz numbers are given by:
	\begin{gather*}
		\frac{1}{d!}
		\sum_{k,l=0}^\infty \frac{t^k u^l}{l!}
		\sum_{|\mu^\pm|=d}
		W^{\circ k,l}(\mu^+,\mu^-)
\\
\qquad{}
= \frac{1}{Z_{\mu^+} Z_{\mu^-}}		\operatorname{Coeff}_{\prod_{a=1}^n z_a^{-\mu_a^+ -1} \prod_{b=1}^m z_{n+b}^{\mu_b^- -1}}
		\Bigg[
		(-1)^{n+m-1} \sum_{\text{$(n+m)$-{\rm cycles}}}
		\prod_{i=1}^m B^{{\rm mix}}_{\sigma(i),\sigma(i+1)}
		\Bigg],
	\end{gather*}	
 where
	\begin{equation*}
		B_{i,j}^{{\rm mix}} = \begin{cases}
		\displaystyle	i_{z_i,z_j} \frac{1}{z_i-z_j} & \text{if} \ i< j\leq n \ \text{or} \ n<i< j,\\
		\displaystyle	\sum\limits_{k=0}^\infty
			\bigg( {\rm e}^{-\frac{k(k+1)}{2}u}\cdot \prod\limits_{j=1}^k \frac{1}{1+jt} \bigg)
			z_i^{-k-1} z_j^k,
			& \text{if} \ i\leq n<j,\\
		\displaystyle	i_{z_j,z_i} \frac{1}{z_i-z_j}, & \text{if} \ j<i\leq n \ \text{or} \ n<j< i,\\
		\displaystyle	-\sum\limits_{k=0}^\infty
			\bigg( {\rm e}^{\frac{k(k+1)}{2}u}\cdot \prod\limits_{j=1}^k \frac{1}{1-jt} \bigg)
			\cdot z_j^{-k-1} z_i^k, & \text{if} \ j\leq n<i.
		\end{cases}
	\end{equation*}
\end{Theorem}

\begin{Example}
	For $(n,m)=(1,1)$, one has
	\begin{equation*}
		\begin{split}
			&\frac{1}{d!}
			\sum_{k,l=0}^\infty \frac{t^k u^l}{l!}
			W^{\circ k,l}((d),(d))
			= \frac{1}{d^2} \sum_{a=0}^{d-1}
			\frac{{\rm e}^{\frac{u}{2}d(d-1-2a)}}
			{\prod_{i=1}^{d-1-a}(1-it) \cdot \prod_{j=1}^a (1+jt)}.
		\end{split}
	\end{equation*}
	For example,
	\begin{gather*}
			 \sum_{k,l=0}^\infty \frac{t^k u^l}{l!}
			W^{\circ k,l}((1),(1)) =1,\\
			 \sum_{k,l=0}^\infty \frac{t^k u^l}{l!}
			W^{\circ k,l}((2),(2)) =\half \bigg(
			\frac{{\rm e}^u}{1-t} + \frac{{\rm e}^{-u}}{1+t}
			\bigg),\\
			 \sum_{k,l=0}^\infty \frac{t^k u^l}{l!}
			W^{\circ k,l}((3),(3)) = \frac{2}{3} \bigg(
			\frac{{\rm e}^{3u}}{(1-t)(1-2t)} + \frac{1}{(1-t)(1+t)}
			+ \frac{{\rm e}^{-3u}}{(1+t)(1+2t)}
			\bigg).
	\end{gather*}
\end{Example}

\begin{Example}
	For $(n,m)=(2,1)$, one has
	\begin{gather*}
			\frac{1}{(a+b)!}
			\sum_{k,l=0}^\infty \frac{t^k u^l}{l!}
			W^{\circ k,l}((a,b),(a+b))
			= \frac{1}{(1+\delta_{a,b})ab(a+b)}\\
			\qquad{}\times
			\sum_{c=0}^{a-1}\Bigg(
			\frac{{\rm e}^{-\frac{u}{2}(a+b)(a-b-1-2c)}}{\prod_{i=1}^{b+c}(1-it) \prod_{j=1}^{a-1-c}(1+jt)}
			-\frac{{\rm e}^{\frac{u}{2}(a+b)(a-b-1-2c)}}{\prod_{i=1}^{a-1-c}(1-it) \prod_{j=1}^{b+c}(1+jt)}
			\Bigg).
	\end{gather*}
\end{Example}

\subsection{Reduction to single Hurwitz numbers}\label{sec-singleH}

Recall that the single Hurwitz numbers can be obtained by taking $\mu^- = (1,1,\dots,1)$
in double Hurwitz numbers labeled by two partitions $\mu^+$ and $\mu^-$.
Thus by evaluating the time variables $\bm t^-$ at $\bm t^- = (1,0,0,0,\dots)$
in the above generating series $\tau^{(2)} (\bm t^\pm;\beta )$, $\tau^{(r)} (\bm t^\pm;\beta )$,
and $W^\bullet (t,u;\bm t^\pm)$ of disconnected double Hurwitz numbers,
one obtains the generating series of the disconnected single Hurwitz numbers.
We denote them by
$\tilde\tau^{(2)} (\bm t^+;\beta )$, $\tilde\tau^{(r)} (\bm t^+;\beta )$
and $\widetilde W^\bullet (t,u;\bm t^+)$, respectively.
Now by Proposition~\ref{prop-restri-t-=1}, we have
\begin{Theorem}
	The KP-affine coordinates for $\tilde\tau^{(r)}(\bm t^+;\beta)$ are
	\[	
 \tilde a^{(r)}_{n,m} =
	\frac{(-1)^n \cdot (m+n)!}{(m+n+1)\cdot m!\cdot n!}
	\cdot \exp\bigg[
	\frac{\beta}{r} \bigg( \bigg(-m-\half\bigg)^r-\bigg(n+\half\bigg)^r
	\bigg)\bigg],
	\qquad m,n\geq 0,
	\]
 for every $r\geq 2$,
	and the KP-affine coordinates for $\widetilde W^\bullet (t,u;\bm t^+)$ are
	\begin{align*}
		\tilde a_{n,m}^{\widetilde W} ={}& (-1)^n \cdot
		\frac{(m+n)!}{(m+n+1)\cdot m!\cdot n!} \\
		&{}\times \exp\Bigg[
		\frac{u}{2}(m+n+1)(m-n)
		- \log\Bigg( \prod_{j=1}^m(1-jt) \prod_{j=1}^n(1+jt)
		\Bigg) \Bigg],
		\qquad m,n\geq 0.
	\end{align*}
\end{Theorem}

Then one can apply Zhou's formula~\cite[Section~5]{zhou1}
to compute the generating series of the connected single Hurwitz numbers.

\section[Stationary Gromov--Witten invariants of P\^{}1 relative to 0,infty]{Stationary Gromov--Witten invariants of $\boldsymbol{\bP^1}$ relative to $\boldsymbol{0,\infty}$}
\label{sec-app-P1}

In~\cite{op},
Okounkov and Pandharipande have studied the Gromov--Witten theory of $\bP^1$
using the Gromov--Witten/Hurwitz correspondence.
Now in this section,
we discuss how to apply the main result in Section~\ref{sec-conn-nmpt} to compute
the stationary GW invariants of $\bP^1$ relative to two points~$0,\infty \in \bP^1$.

Let $(x_1,x_2,\dots)$ be a family of formal variables.
Denote
\begin{equation*}
	\tau_{\bP^1} (x,\bm t^+, \bm t^- )
	=\exp \Bigg(\sum_{|\mu^+|=|\mu^-|}
	\Bigg\langle \mu^+,
	\exp\Bigg(\sum_{i=1}^\infty x_i\tau_i(w)\Bigg),
	\mu^- \Bigg\rangle^{\bP^1}
	\cdot t_{\mu^+}^+ t_{\mu^-}^- \Bigg)
\end{equation*}
the exponential generating functions of the Gromov--Witten invariants,
where $t_\mu=t_{\mu_1} t_{\mu_2} \cdots t_{\mu_l}$ for a partition $\mu = (\mu_1, \mu_2,\dots,\mu_l)$,
and $w\in H^*(\bP^1,\bC)$ is the Poincar\'e dual of $[\text{pt}]$, and
\begin{equation*}
	\Bigg\langle \mu^+ ,\prod_{i=1}^n \tau_{k_i}(w), \mu^- \Bigg\rangle^{\bP^1}
	=\int_{\overline{\mathcal{M}}_{g,n}(\bP^1,\mu^+,\mu^-)} \prod_{i=1}^n \psi_i^{k_i}\operatorname{ev}_i^*(w)
\end{equation*}
is the stationary Gromov--Witten invariants of $\bP^1$
relative to two points $0, \infty \in \bP^1$.
Using the GW/H correspondence,
Okounkov and Pandharipande proved that
(see~\cite[Proposition 4.1]{op})
\[
\tau_{\bP^1} (x,\bm t^+,\bm t^-)
={\rm e}^{ \sum_{i\geq 0} \frac{(1-2^{-i-1}) \zeta(-i-1)}{i+1}x_i}
\cdot \big\langle 0\big| \Gamma_+(\bm t^+) {\rm e}^{\sum_{i\geq 0} x_i K^{(i+1)}} \Gamma_-(\bm t^-) \big|0\big\rangle,
\]
where $\zeta$ is the Riemann zeta-function and $K^{(i+1)}$ are the operators~\eqref{eq-def-C&J-r}.

\begin{Remark}
	The additional factor ${\rm e}^{ \sum_{i\geq 0} \frac{(1-2^{-i-1}) \zeta(-i-1)}{i+1}x_i}$ appears from
	the definition of shifted symmetric power sum.
	It becomes a constant summand after taking logarithm,
	thus makes no contribution to the connected $(n,m)$-point functions.
\end{Remark}

Now in this case,
the corresponding function~\eqref{eq-def-functionf} should be
\[
f_{\bP^1} (s) =
\sum_{i\geq 0} x_i \cdot \frac{s^{i+1}}{(i+1)!},
\qquad \forall s\in \bZ+\half,
\]
then we are able to compute the
stationary GW invariants of $\bP^1$ relative to $0, \infty$
using Theorem~\ref{thm-mainthm-conn}.
The result is
\begin{Theorem}
	Let
	\begin{gather*}
			\mu^+ = (\mu_1^+,\dots,\mu_n^+) = \big(1^{m_1(\mu^+)}2^{m_2(\mu^+)}\cdots\big),\\
			\mu^- = (\mu_1^-,\dots,\mu_m^-) = \big(1^{m_1(\mu^-)}2^{m_2(\mu^-)}\cdots\big),
	\end{gather*}
	be two partitions of integers with $|\mu^+| = |\mu^-|$,
	then the stationary GW invariants of $\bP^1$ relative to $0, \infty \in \bP^1$
	are given by
	\begin{gather*}
		\Bigg\langle \mu^+,
		\exp\Bigg(\sum_{i=1}^\infty x_i\tau_i(w)\Bigg),
		\mu^- \Bigg\rangle^{\bP^1}
		=
		\frac{1}{\prod\limits_{i\geq 1} m_i(\mu^+)! \cdot \prod\limits_{j\geq 1} m_j(\mu^-)!}\\
 \qquad{}\times
		\operatorname{Coeff}_{\prod_{a=1}^n z_a^{-\mu_a^+ -1} \prod_{b=1}^m z_{n+b}^{\mu_b^- -1}}
		\Bigg[
		(-1)^{n+m-1} \sum_{\text{$(n+m)$-{\rm cycles}}}
		\prod_{i=1}^m B^{\bP^1}_{\sigma(i),\sigma(i+1)}
		\Bigg],
	\end{gather*}
 where $B^{\bP^1}_{i,j}$ are given by
	\begin{equation*}
		B_{i,j}^{\bP^1} = \begin{cases}
	\displaystyle		i_{z_i,z_j} \frac{1}{z_i-z_j} & \text{if} \ i< j\leq n \ \text{or} \ n<i< j,\\
	\displaystyle		\sum\limits_{k= 0}^\infty \exp\bigg( -\sum\limits_{l\geq 0} x_l\frac{(k+\half)^{l+1}}{(l+1)!} \bigg) z_i^{-k-1}z_j^k
			& \text{if} \ i\leq n<j,\\
	\displaystyle		i_{z_j,z_i} \frac{1}{z_i-z_j} & \text{if} \ j<i\leq n \ \text{or} \ n<j< i,\\
	\displaystyle		-\sum\limits_{k= 0}^\infty \exp\bigg( \sum\limits_{l\geq 0} x_l\frac{(-k-\half)^{l+1}}{(l+1)!} \bigg) z_j^{-k-1}z_i^k
			& \text{if} \ j\leq n<i.
		\end{cases}
	\end{equation*}
\end{Theorem}

\begin{Example}
	The connected $(1,1)$-point correlators are given by
	\begin{gather*}
			\Bigg\langle (u),
			\exp\Bigg(\sum_{i=1}^\infty x_i\tau_i(w)\Bigg),
			(u) \Bigg\rangle^{\bP^1} \\
\qquad{} =
			\sum_{a=0}^{u-1}
			\exp\bigg( \sum_{k\geq 0} \frac{x_k}{(k+1)!}\bigg( \bigg(a+\half-u\bigg)^{k+1} - \bigg(a+\half\bigg)^{k+1} \bigg) \bigg)\\
 \qquad{} =\sum_{a=0}^{u-1}
			\exp\bigg( f_{\bP^1}\bigg(a+\half -u\bigg) - f_{\bP^1}\bigg(a+\half\bigg) \bigg),
	\end{gather*}
	and the connected $(2,1)$-point correlators are given by
	\begin{gather*}
			\Bigg\langle (u,v),
			\exp\Bigg(\sum_{i=1}^\infty x_i\tau_i(w)\Bigg),
			(u+v) \Bigg\rangle^{\bP^1}\\
			\qquad {}= \frac{1}{1+\delta_{u,v}}
			\sum_{a=0}^{u-1} \bigg[
			- \exp\biggl( \sum_{k\geq 0} \frac{x_k}{(k+1)!}
			\biggl( \biggl(a+\half-u\biggr)^{k+1} - \biggl(a+\half +v \biggr)^{k+1} \biggr) \biggr)\\
			\qquad\quad{} + \exp\biggl(
			\sum_{k\geq 0} \frac{x_k}{(k+1)!}
			\biggl( \biggl(-a-\half-v\biggr)^{k+1} - \biggl(u-\half-a \biggr)^{k+1} \biggr) \biggr)
			\bigg]\\
			\qquad{}= \frac{1}{1+\delta_{u,v}}
			\sum_{a=0}^{u-1} \biggl(
			- \exp
			\biggl( f_{\bP^1}\biggl(a+\half-u\biggr) - f_{\bP^1}\biggl(a+\half +v \biggr) \biggr) \\
			 \qquad\quad{} + \exp
			\biggl( f_{\bP^1}\biggl(-a-\half-v\biggr) - f_{\bP^1}\biggl(u-\half-a \biggr) \biggr)
			\biggr).
	\end{gather*}
\end{Example}

\begin{Remark}
	For a fixed $r$,
	if one takes $x_{r-1} = \beta$ and $x_i=0$ for all $i\neq r$,
	then the function~$f_{\mathbb{P}^1}$ is reduced to $f^{(r)}$
	in the case of double Hurwitz numbers with completed $r$-cycles.
	These are indeed simple cases of GW/H correspondence.
\end{Remark}

\subsection*{Acknowledgements}
We thank the anonymous referees for helpful suggestions.
We also thank Professor Jian Zhou for the on-line course on Hurwitz numbers in TMCSC,
and thank Professor Huijun Fan, Professor Xiaobo Liu and Professor Xiangyu Zhou for encouragement.
The second author is supported by the National Natural Science Foundation of China (No.~12288201).

\pdfbookmark[1]{References}{ref}
\LastPageEnding

\end{document}